\RequirePackage{fix-cm}
\documentclass{svjour3}                     % onecolumn (standard format)
\smartqed  % flush right qed marks, e.g. at end of proof
\usepackage[pdftex]{graphicx}
\usepackage{amsmath}
\usepackage{amsfonts}
\usepackage{amssymb}
\usepackage{color}
\usepackage{comment}
\usepackage{url}
\usepackage{program}
\usepackage{xspace} 
\usepackage{sidecap}
\usepackage{framed}
\usepackage{hyperref}
\newcommand{\fmsolve}{|FM-Solve|$_{\He}$\xspace}

\newcommand{\euf}{\ensuremath{\mathrm{EUF}}\xspace}
\newcommand{\fcc}{\ensuremath{\mathrm{FCC}}\xspace}
\newcommand{\card}[1][]{\mathsf{card}_{#1}}
\newcommand{\So}{\mathsf{S}}
\newcommand{\sat}{\textsf{sat}\xspace}
\newcommand{\unsat}{\textsf{unsat}\xspace}

\newtheorem{thm}{Theorem}

\newtheorem{lem}{Lemma}

\newcommand{\cvc}{\textsc{cvc}{\small 4}\xspace}
\newcommand{\cvciii}{\textsc{cvc}{\small 3}\xspace}
\newcommand{\ziii}{\textsc{z}{\small 3}\xspace}

\newcommand{\dpllt}{DPLL$(T_1)$\xspace}

\newcommand{\dplltsfcconly}{DPLL$(T_\fcc)$\xspace}

\newcommand{\dpllts}{DPLL$(T_1,\ldots,T_m)$\xspace}

\newcommand{\checkfcc}{|check|$_{\fcc}$\xspace}

\newcommand{\rem}[1]{\textcolor{red}{[#1]}}

\newcommand{\ar}[1]{\rem{#1 --ar}}

\renewcommand{\vec}[1]{\mathbf{#1}}

\newcommand{\next}[1]{\mathsf{next}_{#1}}

\newcommand{\instans}[1]{I_{#1}}

\newcommand{\teq}{\approx}
\newcommand{\tneq}{\not\approx}

\newcommand{\cc}[1]{#1^*}
\newcommand{\ccc}[1]{#1^{*}_{c}}

\newcommand{\terms}{\mathbf{T}}

\newcommand{\vals}{\mathbf{V}}

\newcommand{\sfuns}[1]{#1^\mathrm{f}}

\newcommand{\I}{\mathcal{I}}

\newcommand{\M}{\mathcal{M}}

\newcommand{\Mo}{\mathbf{I}}
\newcommand{\He}{\mathcal{H}}

\newcommand{\ent}[1][]{\models_{#1}}

\newcommand{\pent}{\ent[\mathrm{p}]}
\newcommand{\lits}[1]{\mathrm{Lit}_{#1}}
\newcommand{\ilits}[1]{\mathrm{Int}_{#1}}
\newcommand{\none}{\mathsf{no}}
\newcommand{\lev}{\mathsf{lev}}
\newcommand{\limplies}{\Rightarrow}
\newcommand{\liff}{\Leftrightarrow}

\newcommand{\bool}{\mathsf{Bool}}
\newcommand{\distinct}{\mathsf{distinct}}

\newcommand{\compl}[1]{\overline{#1}}

\newcommand{\state}[1]{\langle #1 \rangle}
\newcommand{\Mc}{M}
\newcommand{\Fc}{F}
\newcommand{\Cc}{C}
\newcommand{\pref}[2]{{#1}^{[#2]}}
\newcommand{\failst}{\mathsf{fail}}

\newcommand{\rulename}[1]{\textsf{{\bfseries #1}}\xspace}
\newcommand{\irulename}[2]{\textsf{{\bfseries #1}}$_#2$\xspace}

\newcommand{\backjump}{\rulename{Backjump}}
\newcommand{\decide}{\rulename{Decide}}
\newcommand{\fail}{\rulename{Fail}}
\newcommand{\explain}[1]{\irulename{Explain}{#1}}
\newcommand{\propagate}[1]{\irulename{Propagate}{#1}}
\newcommand{\conflict}[1]{\irulename{Conflict}{#1}}
\newcommand{\plearn}{\rulename{Learn}}
\newcommand{\learn}[1]{\irulename{Learn}{#1}}

\newcommand{\drule}[2]{
\renewcommand{\arraystretch}{1.2}
\(\begin{array}{c}
#1 \\
\hline
#2
\end{array}\)
}

\def\LOOPFOR{\qtab\keyword{for}\ }
\def\LOOPDO{\qtab\keyword{do}\ }
\def\ENDLOOP{\untab}
\def\ENDLOOPDO{\untab\keyword{while}\ }

\def\IF{\qtab\keyword{if}\ }
\def\THEN{\ \keyword{then}\ }
\def\ELSE{\untab\qtab\keyword{else}\ }
\def\ELSEIF{\untab\qtab\keyword{else if}\ }
\def\FI{\untab} %\keyword{end}}
\def\RETURN{\keyword{return}\ }
\def\ENDPROC{\untab}

\def\keyword#1{\mbox{\normalshape\bf #1}}
\def\MATCH{\qtab\keyword{match}\ }
\def\WITH{\ \keyword{with}\ }
\def\ENDMATCH{\untab}

\newcommand{\maxf}{\mathsf{max}}

\newcommand{\ibracket}[1]{[ \negthinspace [ #1 ] \negthinspace ]}
\newcommand{\evaluate}[2]{#2 \ibracket{#1}}
\newcommand{\tupl}[1]{\mathbf{#1}}

\newcommand{\evmap}[1]{\mathcal{V}_{#1}}

\newcommand{\purify}[1]{\lfloor #1 \rfloor}

\journalname{myjournal}
\begin{document}

\title{Constraint Solving for Finite Model Finding in SMT Solvers\thanks{
%Grants or other notes
The work of the first two authors was partially funded 
by a grant from Intel Corporation.
%\rem{more grants?}
}
}


\author{
Andrew Reynolds \and 
Cesare Tinelli \and 
Clark Barrett
}


\institute{
Andrew Reynolds \at
Department of Computer Science, 
The University of Iowa \\
\email{andrew.j.reynolds@gmail.com}           %  \\
%\emph{Present address:} of F. Author  %  if needed
\and
Cesare Tinelli \at
Department of Computer Science, 
The University of Iowa \\
\email{cesare-tinelli@uiowa.edu}
\and
Clark Barrett \at
Department of Computer Science, 
Stanford University \\
\email{barrett@cs.stanford.edu}
}

\date{Received: date / Accepted: date}

\maketitle

\begin{abstract}
SMT solvers have been used successfully as reasoning engines for automated verification and other applications based on automated reasoning. Current techniques for dealing with quantified formulas in SMT are generally incomplete, forcing SMT solvers to report ``unknown'' when they fail to prove the unsatisfiability of a formula with quantifiers. This inability to return counter-models limits their usefulness in applications that produce queries involving quantified formulas. In this paper, we reduce these limitations by integrating finite model finding techniques based on constraint solving into the architecture used by modern SMT solvers. This approach is made possible by a novel solver for cardinality constraints, as well as techniques for on-demand instantiation of quantified formulas. Experiments show that our approach is competitive with the state of the art in SMT, and orthogonal to approaches in automated theorem proving.

\ 

\noindent
Under consideration for publication in Theory and Practice of Logic Programming (TPLP)
\end{abstract}

%==============================================================================
\section{Introduction} \label{sec:intro}
%==============================================================================

Satisfiability modulo theories (SMT) is a subfield of automated reasoning
concerned with the problem of determining the satisfiability of formulas 
in some first-order theory $T$,
where $T$ is usually a combination of several sub-theories. 
SMT techniques and solvers have been used successfully in recent years to support 
a variety of formal methods for hardware and software development,
including automated verification.
They are especially effective for tasks that can be reduced 
to proving the unsatisfiability of \emph{quantifier-free} formulas 
in certain theories, such as theories of linear arithmetic, algebraic datatypes,
bit vectors, arrays, strings and so on, for which it is possible to build
specialized constraint solvers.
A number of applications, however, require reasoners that can prove the 
unsatisfiability of \emph{quantified} formulas in those theories.
In verification applications, for instance, quantified formulas are necessary 
to express properties of systems with an unbounded number of processes, 
or properties involving a number of memory locations.
In general, the need for quantified formulas arises naturally 
when dealing with function or predicate symbols that do not
belong to the signature of an SMT solver's built-in theories.

The few SMT solvers that can currently reason about quantified formulas are based 
on incomplete methods and often report ``unknown'' when they fail, 
after some predetermined amount of effort, 
to prove a quantified formula unsatisfiable.
For many client applications, however, it is very useful to know 
if the failure is due to the fact that the input formula is indeed satisfiable,
especially if the solver can also return some representation 
of the formula's model.
Current SMT solvers are able to produce models of satisfiable quantified formulas
only in fairly restricted cases~\cite{GeDeM-CAV-09}, 
which limits their scope and usefulness.

We reduce these limitations with a novel approach for model finding in SMT.
Since, by the undecidability of first-order logic,
there are no automated methods for finding arbitrary models,
we focus on \emph{finite} models, 
which can be represented symbolically and enumerated.
More precisely, because SMT solvers work with sorted logics
with both built-in and \emph{free} 
(``uninterpreted'') sorts,
we focus on finding models that interpret the free sorts as finite domains.
As with traditional finite model finders for first-order logic,
the main idea is simply to check universally quantified formulas exhaustively 
over candidate models with increasingly large domains for the free sorts,
until an actual model is found.
Our approach differs from previous ones in that it does not rely on the explicit 
introduction of \emph{domain constants} for the free sorts, 
as done by MACE-style model finders~\cite{Claessen:Soerensson:MACEimprove:ModelComputationWS:2003},
and
in that we are able to reason modulo more theories than just the theory of equality,
unlike SEM-style model finders~\cite{Zhang1995IJCAI}.
Moreover, and crucially for our goals, the approach is fully integrated into the general architecture underlying most SMT solvers.

While limited to SMT formulas with quantifiers
ranging only over free sorts or built-in finite sorts (such as, for instance,
bit vector sorts or enumeration sorts), our approach is still quite useful.
Formulas with such quantifiers occur often,
for instance, in verification applications;
moreover, when they are satisfiable they usually have small finite models.

We present our model finding method in the context 
of an abstract framework that models a large class of SMT solvers 
supporting multiple theories and quantified formulas~\cite{KrsGoe-FROCOS-07}.
We incorporate in this framework
an efficient mechanism for deciding the satisfiability of a set 
of ground SMT formulas under \emph{finite cardinality} constraints for the free sorts.
This is used first to find a candidate model,
a model $\M$ of a heuristically generated finite set of ground consequences 
of the input formula $\varphi$.
To check that $\M$ satisfies $\varphi$ as well, the model finder then 
verifies, in a complete way, that \emph{all} ground consequences of $\varphi$ 
over the universe of $\M$ are satisfied by $\M$.
When this check fails, the model finder looks for a new candidate model,
possibly under extended cardinality bounds for the free sorts.
The practical effectiveness of this approach relies on two crucial components:
$(i)$ a method for constructing and representing candidate models efficiently
and 
$(ii)$ a model-based quantifier instantiation heuristic that avoids the explicit generation and checking of all the ground instances of the input formula.
The two are strictly related since the instantiation heuristic
takes advantage of the way candidate models are represented 
to identify, and ignore, entire sets of instances that do not need 
to be considered.

%An overview of this framework is provided in Section~\ref{sec:smt}.
%The method itself is described in Section~\ref{sec:finite-models}.
%In Section~\ref{sec:experimental}, we discuss the initial experimental results 
%obtained with our implementation of the method within the SMT solver \cvc.
The paper is organized as follows.
After discussing preliminaries in Section~\ref{sec:prelim},
we present the framework used by SMT solvers in Section~\ref{sec:dpllts}.
We then present a high-level overview of our approach for finite model finding in Section~\ref{sec:fmf},
followed by details in Sections~\ref{sec:t-fcc}--\ref{sec:fmf-mbqi}
This includes, in particular, the strategy used by the solver 
for finding small candidate finite models and the algorithm for checking 
the satisfiability of quantified formulas against these candidate models.
Section~\ref{sec:results} describes an experimental evaluation 
of our implementation of these techniques in the SMT solver \cvc
on several sets of benchmarks.

This paper builds on material from previous conference papers~\cite{ReyEtAl-1-RR-13,ReyEtAl-2-RR-13},
as well as the PhD dissertation by the first author~\cite{reynolds2013finite}.

\section{Related Work}

Most traditional finite model finders for quantified formulas are based 
on a reduction to a decidable logic,
propositional logic or some decidable fragments of first-order logic, 
where the reduction introduces finite upper bounds on
the cardinalities of the atomic types. This technique was pioneered by McCune
in the Mace tool \cite{mccune-1994}, and is often referred to as MACE-style model finding.
These techniques were later implemented in the tool Paradox \cite{Claessen:Soerensson:MACEimprove:ModelComputationWS:2003},
which incorporated successful techniques for symmetry breaking.
Other conceptually similar tools include FM-Darwin \cite{baumgartner-et-al-2009}, which handles first-order logic with
equality, the Alloy Analyzer with its backend Kodkod \cite{TorJac-TACAS-07}
which handles first-order relational logic, and Nitpick \cite{blanchette-nipkow-2010} which handles higher-order logic.
Recently, a MACE-style finite model finding approach was also implemented in the Vampire theorem prover~\cite{DBLP:conf/sat/Reger0V16}.

A different approach to model finding, 
pioneered by the SEM model finder~\cite{Zhang1995IJCAI}, 
does not encode the input problem into propositional logic.
Instead, it uses built-in support for equality together with constraint propagation techniques 
similar to those used in modern constraint solvers.
In this respect, our approach is more similar to
SEM-style model finding than it is to MACE-style model finding.

Our approach uses on-demand quantifier instantiation to check the satisfiability of universally quantified formulas.
Other instantiation-based approaches have been developed, both in the automated
theorem proving community~\cite{Kor08-IJCAR} and in SMT. 
For the latter, instantiation-based techniques are most typically used 
in an incomplete way 
for finding proofs of unsatisfiability~\cite{Detlefs03simplify:a,GBT09,DBLP:conf/cade/MouraB07}.
Other techniques establish the satisfiability of 
%formulas in the presence of universally 
quantified formulas,
either by using model-based techniques~\cite{GeDeM-CAV-09},
or by reasoning in local theories where only a finite set of instances is required 
for completeness~\cite{ihlemann2008}.

\section{Preliminaries}
\label{sec:prelim}

We work in the context of many-sorted first-order logic 
with equality. %~\cite{Hen-01,Man-MSL-93}.
A (many-sorted) \emph{signature} $\Sigma$ consists of 
a set of sort symbols and
a set of \emph{(sorted) function symbols},
$f : S_1 \times \cdots \times S_n \rightarrow S$,
where $n \geq 0$ and $S_1, \ldots, S_n, S$ are sorts in $\Sigma$.
When $n$ is 0, 
$f$ is also called a \emph{constant symbol}.
We use the binary predicate $\teq$ to denote equality.
We assume that $\Sigma$ always includes a Boolean sort $\bool$ and 
constants $\true$ and $\false$ of that sort---allowing us to encode all other predicate symbols as
function symbols of return sort $\bool$.
For such symbols, we may write, e.g., $P( t_1, \ldots, t_n )$ as shorthand for the equality $P( t_1, \ldots, t_n ) \teq \true$.
%For each sort $S$, we assume our signature contains the if-then-else function symbol $ite : \bool \times S \times S \rightarrow S$.
A signature $\Sigma_0$ is a \emph{subsignature} of a signature $\Sigma$,
and  $\Sigma$ is a \emph{supersignature} of $\Sigma_0$,
if every sort and function symbol of $\Sigma_0$ is also in $\Sigma$.

Given a signature $\Sigma$,
a $\Sigma$-term is either a (sorted) variable $x$ with sort from $\Sigma$,
or an expression of the form $f( t_1, \ldots, t_n )$, 
where $f$ is a function from $\Sigma$, and $t_1, \ldots, t_n$ are $\Sigma$-terms.
A term $t$ is a \emph{well-sorted} term of sort $S$ if $t$ is a variable having sort $S$, 
or $t$ is of the form $f( t_1, \ldots, t_n )$ where
$f$ is of sort $S_1 \times \cdots \times S_n \rightarrow S$, and
$t_1, \ldots, t_n$ are well-sorted terms of sorts $S_1, \ldots, S_n$ respectively.
An \emph{atomic $\Sigma$-formula} is an equality $t_1 \teq t_2$ where $t_1$ and $t_2$ are well-sorted $\Sigma$-terms of the same sort.
A \emph{$\Sigma$-literal} is either an atomic $\Sigma$-formula $p$ or its negation $\neg p$.
We write $s \tneq t$ as an abbreviation for $\neg s \teq t$.
A \emph{$\Sigma$-clause} is a disjunction of $\Sigma$-literals, e.g. $l_1 \vee \ldots \vee l_n$.
%We will use the symbol $\bot$ to denote an empty disjunction of literals.
A \emph{$\Sigma$-formula} is an expression built from atomic $\Sigma$-formulas
logical connectives such as $\vee$, $\wedge$, and $\neg$, and
quantifiers $\forall$ and $\exists$.
An occurrence of variable $x$ is \emph{free} in a formula $\varphi$
if it does not reside within a subformula $\forall x\, \psi$ or $\exists x\, \psi$ of $\varphi$.
We write $FV( \varphi )$ to denote the set of variables that occur free in $\varphi$,
or the \emph{free variables} of $\varphi$.
A \emph{$\Sigma$-sentence} is a $\Sigma$-formula with no free variables.
A $\Sigma$-term or formula is \emph{ground} if it contains no variables.
More generally, by a slight abuse of terminology, we will sometimes call ground 
any quantifier-free term or formula.
%\cb{Do you mean any \emph{set containing} quantifier-free terms or formulas?}
%\ct{fixed}
Where $\tupl{x} = (x_1,\ldots,x_n)$ is a tuple of sorted variables,
we write $\forall \tupl{x}\, \varphi$ as an abbreviation for
$\forall x_1 \cdots \forall x_n\, \varphi$
and $\exists \tupl{x}\, \varphi$ as an abbreviation for
$\neg \forall \tupl{x}\, \neg \varphi$.
When using this notation, we will implicitly assume that $\tupl x$ is maximal---for example
we assume that $\forall x_1\, \forall x_2\, \varphi$ is instead written as $\forall x_1 x_2\, \varphi$.
%\cb{This doesn't make sense to me - what about alternating quantifiers?}
%\ct{see new sentence in parentheses}

%A $\Sigma$-formula is \emph{universal} if it has 
%the form $\forall \tupl{x}\, \varphi$ where $\varphi$ is 
%a quantifier-free formula. % with variables from $\tupl{x}$.

A substitution $\sigma$ is a mapping from variables to terms, applied in postfix form, 
such that $x\sigma$ and $x$ have the same sort for every variable $x$
and the set $\mathcal{D}om(\sigma) := \{ x \mid x \sigma \neq x \}$,
the \emph{domain} of $\sigma$, is finite.
%We say $\sigma$ is a \emph{grounding substitution} (for $\tupl{x}$) 
%if $\sigma$ maps each variable in $\tupl{x}$ to a ground term.
We say $\sigma$ is a \emph{most general unifier (or mgu)} of terms $t_1$ and $t_2$ 
if $\sigma$ is a substitution with minimal domain
such that $t_1 \sigma = t_2 \sigma$.

A \emph{$\Sigma$-interpretation $\I$} maps
each sort $S$ in $\Sigma$ to a non-empty set $S^\I$,
the \emph{domain} of $S$ in $\I$;
it maps each variable $x$ of sort $S$ to an element $x^\I$ of $S^\I$
and each function symbol 
$f : S_1 \times \cdots \times S_n \rightarrow S \in \Sigma$ 
to a total function 
$f^\I : S_1^\I \times \cdots \times S_n^\I \rightarrow S^\I$.
%and each predicate symbol 
%$P : S_1 \times \cdots \times S_n \rightarrow \bool \in \Sigma$ to a subset of
%$S_1^\I \times \cdots \times S_n^\I$.
If $\Sigma_0$ is a subsignature of $\Sigma$,
the \emph{$\Sigma_0$-reduct} of $\I$ is the $\Sigma_0$-interpretation $\I_0$ 
that interprets the symbols of $\Sigma_0$ exactly as $\I$ does.
The evaluation of a term $f( t_1, \ldots, t_n )$ in $\I$, denoted
$\evaluate{t}{\I}$, is defined recursively as
$(i)$ $\evaluate{x}{\I} = x^\I$;
$(ii)$ $\evaluate{f( t_1, \ldots, t_n )}{\I} =$ $f^\I( \evaluate{t_1}{\I}, \ldots \evaluate{t_n}{\I} )$.
%The evaluation of an if-then-else term $ite( \varphi, t_1, t_2 )$ is defined such that
%$\evaluate{ite( \varphi, t_1, t_2 )}{\I} = \evaluate{t_1}{\I}$ if $\evaluate{\varphi}{\I} = \evaluate{\true}{\I}$, and $\evaluate{t_2}{\I}$ otherwise.
%cb Above, you define a substitution as mapping variables to terms of the same
%sort, but below you have a substitution that maps variables to domain
%elements, so there is a slight inconsistency there.
%ajr: inlined discussion of domain elements to avoid talking about two kinds of substitutions.
For a $\Sigma$-interpretation $\I$, a variable $x$ of sort $S$, and an element $u$ 
of  $S^\I$,
we write $\I[x \rightarrow u]$ to denote a $\Sigma$-interpretation that interprets
$x$ as $u$,
and is otherwise identical to $\I$.
The usual satisfiability relation $\models$ between $\Sigma$-interpretations and 
$\Sigma$-formulas, written $\ent$, is defined as follows\footnote{%
Cases for additional constructs such as $\limplies$, $\liff$ and $\exists$ can be defined as usual
by reduction to the cases below.
}
\begin{description}
%\item[-] $\I \ent P(t_1, \ldots, t_n )$ iff $( \evaluate{t_1}{\I}, \ldots, \evaluate{t_n}{\I} )$ $\in \I( P )$
\item[-] $\I \models t_1 \teq t_2$ iff $\evaluate{t_1}{\I} = \evaluate{t_2}{\I}$
\item[-] $\I \models \varphi \wedge \psi$ iff $\I \models \varphi$ and $\I \models \psi$
\item[-] $\I \models \varphi \vee \psi$ iff $\I \models \varphi$ or $\I \models \psi$
\item[-] $\I \models \neg \varphi$ iff $\I \not\models \varphi$
\item[-] $\I \ent \forall x\, \varphi$ iff $\I[x \rightarrow v] \models \varphi$ for all $v \in S^\I$ where $S$ is the sort of $x$.
%cb Is it really enough to say for all grounding substitutions?  Even if the
%language is finite and the domain of the model is infinite?
%ajr: inlined to avoid substitution. 
\end{description}
A $\Sigma$-interpretation $\M$ \emph{satisfies} (or \emph{is a model of}) 
a $\Sigma$-formula $\varphi$ if $\M \models \varphi$;
$\M$ \emph{satisfies} (or \emph{is a model of}) 
a set of $\Sigma$-formulas if it satisfies all of them.
A $\Sigma$-formula or set of $\Sigma$-formulas is \emph{satisfiable} if it has a model
and is \emph{unsatisfiable} otherwise.
We consider only interpretations that interpret $\bool$ as a binary set and 
$\true$ and $\false$ as distinct elements of that set.
We write $\bot$ to abbreviate the unsatisfiable formula $\false \teq \true$.
A set $\Gamma$ of formulas \emph{propositionally entails} a formula $\varphi$, 
written $\Gamma \models_p \varphi$, if the set $\Gamma \cup \{\lnot \varphi\}$ is unsatisfiable
when considering all atomic formulas in it as propositional (Boolean) variables.

A \emph{theory} is a pair $T = (\Sigma, \Mo)$ where
$\Sigma$ is a signature and $\Mo$ is a class of $\Sigma$-interpretations,
the \emph{models} of $T$, closed under variable reassignment
(that is, for all $\I \in \Mo$, every $\Sigma$-interpretation that differs
from $\I$ only in how it interprets variables is also in $\Mo$).
%If $T_0 = (\Sigma_0, \Mo_0)$ is a theory and $\Sigma$ is a supersignature of $\Sigma_0$,
%the \emph{extension of $T_0$ to $\Sigma$} is the theory $T = (\Sigma, \Mo)$ 
%where $\Mo$ is the set of all $\Sigma$-interpretations $\I$
%whose $\Sigma_0$-reduct is a model of $T_0$.
%We refer to the symbols of $\Sigma$ that are not in $\Sigma_0$ as \emph{uninterpreted}
%and all other symbols as \emph{interpreted}.
%A constant symbol $c$ of $\Sigma$ of sort $S$ is \emph{free} in $T$ 
%if it effectively behaves like a variable, that is, for all $\I \in \Mo$, 
%every $\Sigma$-interpretation that differs from $\I$ only on how it interprets
%$c$ is also in $\Mo$.
%
The \emph{union} of two theories $T_1 = (\Sigma_1, \Mo_1)$ and $T_2 = (\Sigma_2, \Mo_2)$,
when it exists,
is the theory $T_1 \cup T_2 = (\Sigma, \Mo)$
where $\Sigma$ is the smallest supersignature of $\Sigma_1$ and $\Sigma_2$
and $\Mo$ is the set of all $\Sigma$-interpretations whose $\Sigma_i$-reduct
is in $\Mo_i$ for $i=1,2$.
This definition extends to more than two theories as expected.
%Note that $T_1 \cup T_2$ might not exist.
%It is easy to show though that when it does it is the same as $T_2 \cup T_1$.
%Similarly, $T_1 \cup (T_2 \cup T_3) = (T_1 \cup T_2) \cup T_3$
%whenever $T_1 \cup T_2$, $T_1 \cup T_3$ and $T_2 \cup T_3$ exist.

Given a theory $T = (\Sigma, \Mo)$, 
a $\Sigma$-formula $\varphi$ is \emph{satisfiable modulo $T$}, or \emph{$T$-satisfiable},
if and only if there is a model of $T$ that satisfies $\varphi$.
A set $\Gamma$ of $\Sigma$-formulas \emph{$T$-entails} a $\Sigma$-formula $\varphi$,
written $\Gamma \ent[T] \varphi$,
if every model of $T$ that satisfies all formulas in $\Gamma$ satisfies $\varphi$ as well.
The formula $\varphi$ is \emph{$T$-valid} if it is $T$-entailed by the empty set---equivalently,
if it is satisfied by every model of $T$.
Two sets $\Gamma_1$ and $\Gamma_2$ of $\Sigma$-formulas are \emph{equisatisfiable in $T$}
if for every model of $T$ that satisfies one there is a model of $T$ that satisfies the other,
and the two models differ at most on the way they interpret the free variables
not shared by $\Gamma_1$ and $\Gamma_2$.

\section{The \dpllts Framework}
\label{sec:dpllts}

Most SMT solvers have a basic architecture that combines in a principled way
a propositional satisfiability solver, the \emph{SAT engine}, 
with a number of \emph{theory} solvers,
specialized constraint solvers for sets of literals over a specific theory. 
A general framework, called DPLL($T$), to describe at an abstract but formal level
the working of such SMT solvers and the interaction of their main components
was originally developed by Nieuwenhuis et al.~\cite{NieOT-JACM-06}.
The framework, parametrized by a background theory $T$, describes entire families
of procedures to determine the $T$-satisfiability of a ground set of input clauses.
We present here a variant of it, introduced by Krsti\'c and Goel~\cite{KrsGoe-FROCOS-07}, 
where $T$ is not a monolithic theory 
but is instead the union of a number of separate sub-theories,
each with its own theory solver.

\subsection{The theory $T$}

For the rest of the paper, we will consider a $\Sigma$-theory $T = T_1 \cup \cdots \cup T_m$
where each $T_i$ is a theory with signature $\Sigma_i$.
One of these theories, say $T_e$ with $e \in \{1,\ldots,m\}$,
may be the theory of equality---over the symbols $\Sigma_e$.
This theory, whose set of models consists of all $\Sigma_e$-interpretations, 
is also known as the theory of equality with uninterpreted functions (\euf).
As a consequence, we will refer to the sort and function symbols of $T$
that occur only in $\Sigma_e$ as \emph{uninterpreted} and 
to the other symbols of $T$ as \emph{interpreted}.
To stress that we treat the component theories of $T$ individually 
we will refer to our variant of DPLL($T$) as \dpllts.

%For convenience and without loss of generality, we assume that 
%every signature of $\Sigma_1, \ldots, \Sigma_m$ shares the same $\sorts$ 
%of sort symbols (including $\bool$) with every other signature,
%\cb{This means that according to the previous sentence, there are no
%  uninterpreted sorts.  Is this intended?}
%CT reply to Clark. The previous sentence did not imply that. Rewritten to clarify.
For convenience and without loss of generality, we assume that 
if a signature from $\{\Sigma_1, \ldots, \Sigma_m\}$ shares a
sort symbol with another signature then it shares it also with all the signatures in the set.
%AJR : this is not used anywhere, and a reviewer complained about it
%Also, where $\sorts$ is the set of all sorts of $T$,
%we fix a distinguished set $\fc_S$ 
%of variables of sort $S$ for each $S \in \sorts$,
%to be used internally by the framework (and not in input formulas).  
%Let $\fc = \bigcup_{S \in \sorts} \fc_S$.
Finally, we impose the (true) restriction that
the signatures $\Sigma_1, \ldots, \Sigma_m$ share no function symbols at all
except for $\true$ and $\false$.
This restriction is currently imposed by all SMT solvers that support
multiple theories as it enables the modular combination of theory solvers
for the individual theories.

\subsection{Transition system}

The \dpllts framework for theory $T$ defines a state transition system
for each ground $\Sigma$-formula $\varphi_0$ whose $T$-satisfiability one is interested in.
Intuitively, the initial state of the system corresponds to a CNF encoding of $\varphi_0$. 
Under the right conditions on $T$, all of the
executions of the system starting from such a state end in a distinguished 
fail state if and only if $\varphi_0$ is not $T$-satisfiable.

%-----------------------------------------------------------------------------
\paragraph{States}
%-----------------------------------------------------------------------------
System states are all triples of the form $\state{M,F,C}$
where
\begin{itemize}
\item
$M$, the current \emph{assignment}, is a sequence of literals and 
\emph{decision points $\bullet$},
\item
$F$ is a set of ground clauses derived from $\varphi_0$, and 
\item
$C$ is either the distinguished value $\none$ or a clause, which we will refer to as a \emph{conflict clause}.
\end{itemize}

\noindent
Each assignment $M$ can be factored uniquely into 
the subsequence concatenation $M_0\bullet M_1 \bullet \cdots \bullet M_n$,
%using juxtaposition for sequence concatenation,
where no $M_i$ contains decision points.
For $i=0,\ldots,n$, we call $M_i$ the \emph{decision level} $i$ of $M$ and
denote with $\pref{M}{i}$ the subsequence $M_0\bullet \cdots \bullet M_i$.
When convenient, we will treat $M$ as the set of its literals.
%and call them the \emph{asserted literals}.
The formulas in $F$ have a particular \emph{purified form}
that can be assumed with no loss of generality
since any formula can be efficiently converted into 
that form while preserving equisatisfiability in $T$:
each element of $F$ is a ground clause,
and each atom occurring in $F$ is \emph{pure},
that is, has signature $\Sigma_i$ for some $i\in \{1, \ldots, m\}$.
By the way assignments are constructed, their atoms too are always pure.

Initial states have the form 
$\state{\emptyset, F_0, \none}$ where $F_0$ is an input set of clauses
to be checked for $T$-satisfiability.
The expected final states are states of the form $\state{M, F, \bot}$, 
when $F_0$ is not $T$-satisfiable;
or else $\state{M, F, \none}$ with 
$M$ satisfiable in $T$,
$F$ equisatisfiable with $F_0$ in $T$, and
$M \pent F$.

If $M$ is $T$-satisfiable and $M \pent F$ 
we call $M$ a \emph{satisfying assignment} for $F$.

\begin{figure}[t]
\propagate{i}
\drule{
l_1, \ldots, l_n \in \Mc \quad 
l_1, \ldots, l_n \ent[i] l \quad
l \in \lits{\Fc} \cup \ilits{\Mc} \quad  
l, \compl{l} \notin \Mc
}
{\Mc := \Mc\;l}
\medskip

\decide
\drule{
l \in \lits{\Fc} \cup \ilits{\Mc} 
\quad 
l, \compl{l} \notin \Mc
}
{\Mc := \Mc \bullet l}
\qquad
\conflict{i}%
\drule{
\Cc = \none \quad 
l_1, \ldots, l_n \in \Mc \quad 
l_1, \ldots, l_n \ent[i] \bot
}
{\Cc := \compl{l}_1 \lor \cdots \lor \compl{l}_n}
\medskip

\explain{i}
\drule{
\Cc = l \lor D \quad 
\compl{l}_1, \ldots, \compl{l}_n \ent[i] \compl{l} \quad 
\compl{l}_1, \ldots, \compl{l}_n \prec_\Mc \compl{l}
}
{\Cc := l_1 \lor \cdots \lor l_n \lor D}
\medskip

%\learn{i}
%\drule{
%\emptyset \ent[i] \exists \tupl{x}\, (l_1[\tupl{x}] \lor \cdots \lor l_n[\tupl{x}]) \quad 
%l_1, \ldots, l_n \in \lits{\Mc}\vert_i \cup \ilits{\Mc} \cup L_i
%}
%{\Fc := \Fc \cup \{l_1[\tupl{c}] \lor \cdots \lor l_n[\tupl{c}]\}}
%\medskip
\learn{i}
\drule{
\emptyset \ent[i] l_1 \lor \cdots \lor l_n \quad 
l_1, \ldots, l_n \in \lits{\Mc}\vert_i \cup \ilits{\Mc} \cup L_i
}
{\Fc := \Fc \cup \{l_1 \lor \cdots \lor l_n \}}
\qquad
\learn{0}
\drule{
\Cc \neq \none \quad
\bullet \in \Mc
}
{\Fc := \Fc \cup \{\Cc\}}
\medskip

\backjump
\drule{
\Cc = l_1 \lor \cdots \lor l_n \lor l
\quad \lev\ \compl{l}_1, \ldots, \lev\ \compl{l}_n \leq\ i < \lev\ \compl{l}
}
{\Cc := \none
\quad
\Mc := \pref{\Mc}{i}\: l
}
\qquad
\fail
\drule{ 
\Cc \neq \none \quad
\bullet \notin \Mc
}
{\Cc := \bot}
%\medskip

\caption{\dpllts rules}
\label{fig:rules}
\end{figure}

%-----------------------------------------------------------------------------
\paragraph{Transition rules}
%-----------------------------------------------------------------------------
The possible behaviors of the system are defined 
by a set of non-deterministic state transition rules,
specifying a set of successor states for each current state.
The rules are provided in Figure~\ref{fig:rules}
in \emph{guarded assignment form}~\cite{KrsGoe-FROCOS-07}.
A rule applies to a state $s$ if all of its premises hold for $s$.
In the rules,
$\Mc$, $\Fc$ and $\Cc$ respectively denote 
the assignment, formula set, and conflict clause component of the current state.
The %rule's 
conclusion describes how each component is changed, if at all.

We write $\compl{l}$ to denote the complement of literal $l$ and write
$l \prec_\Mc l'$ to indicate that $l$ occurs before $l'$ in $\Mc$.
The function $\lev$ maps each literal of $\Mc$ to the (unique) decision level
at which $l$ occurs in $\Mc$.
The set $\lits{\Fc}$ (resp., $\lits{\Mc}$) consists of all literals 
in $\Fc$ (resp., all literals in $\Mc$) and their complements.
For $i=1,\ldots,m$, the set $\lits{\Mc}\vert_i$ consists 
of the $\Sigma_i$-literals of $\lits{\Mc}$.
$\ilits{\Mc}$ is the set of all \emph{interface literals} of $\Mc$:
the equalities and disequalities between \emph{shared variables}
where the set of shared variables is
\[
\{x\ \mid\ x \text{ is a variable in both } 
             \lits{\Mc}\vert_i \text{ and } \lits{\Mc}\vert_j
             \text{for some }1 \le i < j \le m \} \ .
\]
%\cb{There seems to be a problem here.  We define ground as containing no
%  variables, but we seem to need variables for purifying and sharing.}
%\ct{I have redefined ground to mean just quantifier-free. Since we use interpretations
%(as opposed to structures) free variables can be used as free constants.
%This way we do not have to confuse the latter with uninterpreted constants.
%}
%\cb{The definition of ground still says containing no variables.}

The index $i$ in the rules ranges from $0$ to $m$ 
for \propagate{i}, \conflict{i} and \explain{i},
and from $1$ to $m$ for \learn{i}.
In all rules, $\ent[i]$ abbreviates $\ent[T_i]$ when $i > 0$.
In \propagate{0}, $l_1,\ldots, l_n \ent[0] l$ simply means that 
$\compl{l}_1 \lor \cdots \lor \compl{l}_n \lor l \in \Fc$.
Similarly,
in \conflict{0}, $l_1,\ldots, l_n \ent[0] \bot$ means that 
$\compl{l}_1 \lor \cdots \lor \compl{l}_n \in \Fc$;
in \explain{0}, $\compl{l}_1, \ldots, \compl{l}_n \ent[0] \compl{l}$ means that 
$l_1 \lor \cdots \lor l_n \lor \compl{l} \in \Fc$.

The rules \decide, \propagate{0}, \explain{0}, \conflict{0}, %\fail, 
\plearn, and \backjump
model the behavior of the SAT engine, which treats atoms as Boolean variables.
%We call \emph{decision literal} the literal $l$ added to $\Mc$ by \decide.
The rules \conflict{0} and \explain{0} model the conflict discovery and analysis
mechanism used by CDCL SAT solvers.
All the other rules model the interaction between the SAT engine and the individual theory solvers
in the overall SMT solver.

Generally speaking, the system uses the SAT engine to construct
the assignment $\Mc$ as if the problem were propositional, 
but it periodically asks the sub-solvers for each theory $T_i$ to check 
if the set of $\Sigma_i$-literals in $\Mc$ is $T_i$-unsatisfiable, 
or entails in $T_i$ some yet undetermined literal from $\lits{\Fc}\cup \ilits{\Mc}$.
In the first case, the sub-solver returns an \emph{explanation} of 
the unsatisfiability as a conflict clause, which is modeled by \conflict{i}
with $i\in\{1,\ldots,m\}$.
The propagation of entailed theory literals and the extension 
of the conflict analysis mechanism to them is modeled by the rules 
\propagate{i} and \explain{i}.

The inclusion of the interface literals $\ilits{\Mc}$ 
in \decide and \propagate{i} achieves the effect of the Nelson-Oppen
combination method~\cite{BruCFGS-AMAI-09,TinHar-FROCOS-96}.
Under the right conditions on the component theories, 
the two rules allows the overall system to determine the $T$-satisfiability
of the input formula by doing only local reasoning in the individual 
component theories and exchanging information between their corresponding
solvers just through (dis)equalities between interface variables.
% by making 
%it possible to guarantee that any two sub-solvers agree on the equality or
%disequality of every pair of constants they share. 

The rule \learn{i} with $i > 0$ is needed to model theory solvers following 
the splitting-on-demand paradigm~\cite{BarNOT-LPAR-06}.
When asked about the satisfiability of their constraints,
these solvers may instead return a \emph{splitting lemma}, 
a $T_i$-valid formula that encodes an additional guess 
that needs to be made about the literals in $\Mc$
before the solver can determine their satisfiability.
The set $L_i$ in the rule is a finite set consisting of literals, not present 
in the input set $F_0$, which may be generated by such solvers. 

%The \ainst and \einst rules model the quantifier instantiation mechanism.
%When the atom $a$, which serves as a proxy for the quantified formula
%$\forall \tupl{x}\,C$, occurs positively in the current assignment $\Mc$,
%the SMT solver adds one or more ground instances of the clause
%$a \limplies  C[\tupl{x}]$.
%When $a$ occurs negatively, the system adds the (Skolemized) clause form of
%$\lnot a \limplies \lnot\forall \tupl{x}\,C$.
%Instantiation heuristics dictate which instances are generated and 
%how quantifier instantiation applications are interleaved with 
%the other operations.

%-----------------------------------------------------------------------------
\subsection{System executions and correctness}
%-----------------------------------------------------------------------------
An \emph{execution} of a transition system modeled as above
is a (possibly infinite) sequence $s_0, s_1, \ldots$ of states such that 
$s_0$ is an initial state and 
for all $i \geq 0$, $s_{i+1}$ can be generated from $s_i$ by the application 
of one of the transition rules.
A system state is \emph{reachable} if it occurs in some execution;
it is \emph{irreducible} if no transition rules besides \learn{i}, apply to it.
An \emph{exhausted execution} is a finite execution 
whose last state is irreducible.
A \emph{complete execution} is either an exhausted execution or an infinite execution.
An application of \learn{i}, with $i \geq 0$, is \emph{redundant} in an execution 
if the execution contains a previous application of \learn{i}
with the same premise.

Adapting results from~\cite{BarNOT-LPAR-06,KrsGoe-FROCOS-07,NieOT-JACM-06},
it can be shown that every execution 
starting with a state $\state{\emptyset, F_0, \none}$ and
ending in $\state{M,F,C}$ satisfies the following invariants:
\begin{enumerate}
\item $M$ contains only pure literals and no repetitions;
\item $F \ent[T] C$ and $M \pent \lnot C$ when $C \neq \none$;
\item $F_0$ and $F$ are equisatisfiable in $T$.
\end{enumerate}
Moreover, % in the absence of quantified formulas, 
the transition system is \emph{terminating}: 
every execution with no redundant applications of \learn{i} is finite; and
\emph{sound}:
for every execution starting with a state $\state{\emptyset, F_0, \none}$ and 
ending with a state $\state{M, F, \bot}$, 
the clause set $F_0$ is $T$-unsatisfiable.
Under suitable assumptions on the sub-theories $T_1, \ldots, T_m$,
the system is also \emph{complete}:
for every exhausted execution starting with $\state{\emptyset, F_0, \none}$ and 
ending with $\state{M, F, \none}$, $M$ is satisfiable in $T$
and $M \pent F_0$.
%
%With quantified formulas, soundness is preserved 
%but termination and completeness are lost in general.
Here, we provide a sketch of the correctness proof for \dpllts restricted 
to a single theory $T_1$
based on the proof for the original framework~\cite{NieOT-JACM-06}.
%In that case, no restrictions on the theory $T$ are needed
%(see~\cite{KrsGoe-FROCOS-07} for the multiple-theory case).

\begin{thm}
\label{thm:dpllts-correct}
Suppose $T = T_1$.
With any strategy where all applications of \learn{1} are not redundant and
introduce new literals only from a finite set $L_1$,
\dpllt is sound, complete, and terminating for all sets $F_0$ of ground clauses.
\end{thm}
\begin{proof}
(Sketch)

Soundness)
Observe that all reachable states of the form $\state{M,F,C}$ where $C \neq \none$
are such that $C$ is $T_1$-entailed by $F_0$.
When $C$ is introduced by \conflict{0}, it is a clause from $F$;
when it is introduced by \conflict{1}, it is $T_1$-valid.
When applying \explain{i}, we replace a literal $l$ in $C$ 
with disjunction of literals $l_1 \lor \cdots \lor l_n$
which is entailed by $l$ either in the theory $T_1$ (when $i = 1$), 
or together with $F$ (when $i=0$).
Thus, when a state of the form $\state{M,F,\bot}$ is reachable, 
we can conclude  that $F \ent[T_1] \bot$.
Since all clauses in $F$ are $T_1$-entailed by $F_0$ by construction, 
$\bot$ is $T_1$-entailed by $F_0$ as well.
\smallskip

Termination)
For all reachable states $\state{M,F,C}$, every literal occurring in $M$, $F$ or $C$ 
belongs to the finite set of literals $\lits{\Mc} \cup \ilits{\Mc} \cup L_1$.
As a consequence, there is only a finite number of reachable states.
Consider a partial ordering $\succeq$ on assignments $M$, with the empty assignment
as maximal element,
such that $(e_1\, M_1) \succeq (e_2\, M_2)$ 
if either 
$(i)$ $e_1 = \bullet$ and $e_2 \neq \bullet$ or 
$(ii)$ $e_1 = e_2$ and $M_1 \succeq M_2$.
In addition, 
consider a partial ordering $\succeq$ on conflict clauses, seeing as sets of literals, 
such that $C_1 \succeq C_2$ if either $C_1$ is $\none$, 
or neither $C_1$ nor $C_2$ are $\none$ and $C_2 \prec_M^{mul} C_1$, where $\prec_M^{mul}$ is the multiset extension of $\prec_M$.
Extend this ordering to states so that $\state{M_1,F_1,C_1} \succeq \state{M_2, F_2, C_2}$ if and only if
$M_1 \succeq M_2$, or $M_1 = M_2$ and $C_1 \succeq C_2$.
One can show that this ordering is well founded.
Moreover, applying all rules besides \learn{1} to a state $s$ results in a state $s'$ where $s \succ s'$.
Termination then follows from the fact that \learn{1} is applicable
only a finite number of times.
\smallskip

Completeness)
%\ct{The argument here is bogus. We need to assume and use that the component theories
%satisfy Nelson-Oppen-style combination conditions.}
We claim that for every irreducible reachable state $\state{M,F,\none}$,
$M$ is a $T_1$-satisfiable satisfying assignment for $F$.
To see this, consider first that since \decide does not apply to the state,
$M$ must contain an assignment for all literals in $F$.
Moreover, $M$ is a satisfying assignment for $F$ since \conflict{0} does not apply.
Since \conflict{1} does not apply, then $M$ must be $T_1$-satisfiable.
Since $F_0 \subseteq F$, we have that $M$ propositionally entails $F_0$,
and thus $F_0$ is $T_1$-satisfiable as well.
Thus, since our procedure is terminating, it is also complete.
\qed
\end{proof}
\medskip

The soundness and termination arguments in the proof above
immediately extend to the multiple theory case of \dpllts where $m > 1$.
The completeness argument can be extended as well under further model-theoretic 
assumptions on the component theories~\cite{KrsGoe-FROCOS-07,DBLP:journals/fmsd/JovanovicB13}.

%==============================================================================
\subsection{A Typical Strategy for \dpllts}
\label{sec:dpllt-strat}
%==============================================================================

\begin{figure}[t]
\begin{tabular}{l@{\qquad\qquad}l}
\begin{minipage}[t]{.4\linewidth}
\begin{program}
\PROC |check|( \Mc, \Fc, \Cc ) \BODY 
  (\text{\propagate{0}} \mid \ldots \mid \text{\propagate{n}})^\ast;
  \IF |weak\_effort|( \Mc, \Fc, \Cc ) = \mathsf{true} 
    \IF l, \compl{l} \notin \Mc \text{ for some } l \in \lits{F}
      \decide \text{ on } l
    \ELSEIF |strong\_effort|( \Mc, \Fc, \Cc )
      \RETURN \state{ \Mc, \Fc, \none }
    \FI
  \FI
  \RETURN |check\_conflict|( \Mc, \Fc, \Cc )
\ENDPROC
\end{program}
\end{minipage}
&
\begin{minipage}[t]{.4\linewidth}
\begin{program}
\PROC |check\_conflict|( \Mc, \Fc, \Cc ) \BODY 
  \IF \Cc \neq \none
    (\text{\explain{0} } \mid \ldots \mid \text{ \explain{n}})^\ast;
    \IF \Cc = \emptyset
      %\text{\fail};
      \RETURN \state{\Mc, \Fc,\bot}
    \ELSE
      \plearn;\backjump
    \FI
  \FI
  \RETURN |check|(\Mc, \Fc, \Cc)
\ENDPROC
\end{program}
\end{minipage}
\end{tabular}
\caption[A typical strategy $|check|$ for applying \dpllts rules]{
A typical strategy $|check|$ for applying \dpllts rules.
}
\label{fig:dpllt-strategy}
\end{figure}

A typical strategy for applying the theory-specific rules of \dpllts
is outlined in Figure~\ref{fig:dpllt-strategy}
where $\mid$ denotes alternative choice and ${}^*$ denotes 
zero or more rule applications.
% AJR : we do discuss rules with index 0 in the Figure.
%\footnote{%
%We do not discuss the Boolean rules (those with index 0) here
%as they model the behavior of the SAT engine.
%}
%The strategy is broken up into two methods $|check|$ and $|check\_conflict|$.
The $|check|$ procedure involves two sub-procedures $|weak\_effort|$ and $|strong\_effort|$, 
which are not shown here and are specific to the theory $T$.
Each of these methods when invoked either
applies \conflict{i} or \learn{i} for some $1 \leq i \leq m$ and 
returns $\mathsf{false}$, or
applies no rules and returns $\mathsf{true}$.
%Intuition

\emph{Weak effort checks}, as denoted by $|weak\_effort|$,
are commonly used to eagerly avoid extensions that are clearly unsatisfiable in one of the theories.
%\ct{this is unclear}.
On the other hand, \emph{strong effort checks}, as denoted by $|strong\_effort|$,
are required to make progress towards determining the $T$-satisfiability of the conjunction of literals in $M$.
In particular, unlike $|weak\_effort|$, we require that $|strong\_effort|$ returns
$\mathsf{true}$ only when $M$ is $T$-satisfiable.
Generally speaking, weak effort checks typically involve computationally inexpensive reasoning at the cost of incompleteness,
whereas strong effort checks are complete but may involve expensive reasoning.
%Although not shown here, 
%we assume the weak and strong effort checks consist of sequentially calling
%weak and strong effort check methods of the individual theory solvers for $T_1, \ldots, T_n$,
%where we denote $|weak\_effort|_T$ (resp. $|strong\_effort|_T$) as the weak (resp. strong) effort
%check for theory $T$.
The design of a theory solver in the \dpllts framework depends largely on how the methods
$|weak\_effort|$ and $|strong\_effort|$ are implemented.
We will see an example of these functions in Section~\ref{sec:fcc-dpllts-opt}.

In more detail, the first sub-procedure $|check|$ in Figure~\ref{fig:dpllt-strategy} applies to states $\state{M,F,C}$ where $C = \none$.
We first apply the rule \propagate{i} for sub-theories $T_i$, possibly multiple times.
Afterwards, we apply a weak effort check.
If no conflicts or clauses are learned at weak effort, 
we apply \decide on some unassigned literal $l$ from $\lits{F}$, if one exists.
Otherwise, our assignment $M$ is complete, and we apply a strong effort check to verify the $T$-satisfiability of $M$.
If $|strong\_effort|( M, F, C )$ returns $\mathsf{true}$, then $M$ is satisfiable in $T$,
and the method returns the (final) state $\state{ M, F, \none }$,
indicating that $F$ is satisfiable.
In all other cases, we apply $|check\_conflict|$.

The second sub-procedure $|check\_conflict|$ is applied to states $\state{M,F,C}$ where $C$ may be different from $\none$.
In those cases,
% where conflict clause $C$ was discovered (either during a weak or a strong effort check),
we perform conflict analysis by repeated applications of \explain{i}.
If we reach the $\failst$ state, % of the form $\state{M,F,\bot}$, 
then we know $F$ is unsatisfiable.
Otherwise, we add a learned clause via \learn, and apply \backjump to return to a prefix of $M$.
%\cb{Please check: changed from ``we may add a learned clause'' because figure does not allow not applying \learn.}
%\ct{Fix is OK.}

%For soundness and termination, 
%we require that each call to $|weak\_effort|$ and $|strong\_effort|$ legally executes a set of \dpllts rules, each of which are not redundant in the current execution.
%For termination, we require that $|weak\_effort|$ and $|strong\_effort|$ return $\mathsf{false}$ only when they apply at least one rule.
%For completeness, we require that $|strong\_effort|$ returns $\mathsf{true}$ only when 
%the conjunction of the literals in $M$ is $T$-satisfiable.

We formally state the requirements of weak and strong effort checks 
for the single theory case in the following proposition
%\ar{define soundness/termination/completeness for strategy}
which is a consequence of Theorem~\ref{thm:dpllts-correct}.

\begin{proposition}
\label{prop:dpllt-strat}
Suppose $T = T_1$.
The |check| method in Figure~\ref{fig:dpllt-strategy} implements
a sound, complete, and terminating strategy for all sets $F_0$ of ground clauses provided all of the following hold.
\begin{enumerate}
\item 
In $|weak\_effort|$ and $|strong\_effort|$, all applications of \learn{1} 
are not redundant and 
introduce new literals only from a finite set $L_1$.
%The new literals introduced by $|weak\_effort|$ and $|strong\_effort|$ in applications of \learn{1} are taken from a finite set $L_1$, and
\item 
$|weak\_effort|$ and $|strong\_effort|$ %legally execute a set of non-redundant applications of \dpllt rules and
return $\mathsf{false}$ only when they apply at least one rule.
\item $|strong\_effort|( \Mc, \Fc, \Cc )$ returns $\mathsf{true}$ 
only when the conjunction of the literals in $M$ is $T_1$-satisfiable.
\end{enumerate}
\end{proposition}
%\begin{proof}
%By Theorem~\ref{
%\end{proof}
%\ar{Proof is immediate, mention?}
\begin{proof} (Sketch)
The first point ensures that |check| meets the requirements 
on applications of \learn{1} as given in Theorem~\ref{thm:dpllts-correct};
the second point ensures that |check| is a terminating method; and
the third point ensures that |check| generates exhaustive executions.
\qed
\end{proof}

\section{Finite Model Finding in SMT}
\label{sec:fmf}

The \dpllts framework described in the previous section is limited to quantifier-free formulas.
This section outlines an approach for finite model finding for quantified formulas
that can be integrated in \dpllts-based SMT solvers.
%
%In the remainder of the paper, we consider
%the more general setting of establishing the $T$-satisfiability of formulas
%in the presence of universally quantified formulas.
Concretely, we consider $\Sigma$-formulas in the following language:
\[\begin{array}{lcl}
\phi & := & t_1 \teq t_2 \mid \phi_1 \wedge \phi_2 \mid \phi_1 \vee \phi_2 \mid\neg \phi \mid \forall x\, \varphi
\end{array}\]
where $t_1$ and $t_2$ are $\Sigma$-terms, and 
\emph{the sort of $x$ is either an uninterpreted sort or a sort interpreted 
in every model of $T$ as a finite set of some fixed cardinality.}
%Recall that SMT solvers work with sorted logics containing both interpreted and uninterpreted sorts.
%Finite model finding focuses on finding \emph{finite models},
%that is, models that interpret each uninterpreted sort as a finite set.
%The approach is applicable to cases where each universal quantifier in the problem is either 
%over an uninterpreted sort, or a finite interpreted sort.
Examples of the latter include sorts denoting fixed-length bit-vectors
or finite (non-recursive) datatypes.
Certain integer arithmetic constraints with bounded quantifiers,
where the bounds are explicitly provided or can be inferred,
can be treated similarly to finite interpreted sorts~\cite{reynolds2013finite,DBLP:conf/cade/BaumgartnerBW14}.
Many applications of SMT rely on solving problems that fall into such categories, 
given a careful encoding of the constraints.

Given an input formula $\psi$ in the grammar above,
our approach first performs a purification step,
which results in a set $F$ of ground clauses, and 
a set $A$ of equivalences of the form $a \teq \true \liff \forall \tupl{x}\, \varphi$,
abbreviated as $a \liff \forall \tupl{x}\, \varphi$,
where $a$ is a Boolean variable uniquely associated 
with the quantified formula $\forall \tupl{x}\, \varphi$.
We will refer to $a$ as the \emph{proxy variable} for $\forall \tupl{x}\, \varphi$.
%constructs a set of \emph{purified $T$-clauses} $F$,
%where each clause in $F$ is ground, 
%and some Boolean variables, which we refer to as \emph{proxy variables}, are uniquely associated with quantified formulas.
%We write $a \liff \forall \tupl{x}\, \varphi$ to denote that $a$ is the proxy variable for $\forall \tupl{x}\, \varphi$.
The set $F$ can be constructed by a standard conversion of $\psi$ to clausal form
which, however, treats the quantified subformulas of $\psi$ as atoms.
After that, each quantified formula $\forall \tupl{x}\, \varphi$ occurring in a clause of $F$ 
is replaced with its proxy variable $a$ if it occurs positively in the clause, and 
with $\varphi\{ \tupl{x} \mapsto \tupl{k} \}$ otherwise,
where $\varphi\{ \tupl{x} \mapsto \tupl{k} \}$ is the result of 
substituting each occurrence in $\varphi$ of a variable $x$ of $\tupl x$ 
with a fresh variable $k$ (to be treated like a Skolem constant).
The process is repeated until $F$ contains no quantifiers.
%Let $A$ be the conjunction of equivalences of the form $a_i \liff \forall \tupl{x}\, \varphi_i$ for each proxy variable
%$a_i$ occurring in $F$.
This conversion from $\psi$ to $F$ and $A$ can be done so that 
$\psi$ and $F \cup A$ are equisatisfiable.
We will denote by $\purify{\psi}$ the resulting pair $( F, A )$.

\begin{example}
%Let $T$ be the theory of equality and uninterpreted functions.
Consider the formula $\psi =  \neg P( b, c ) \wedge ( Q( b, c ) \liff \forall x\, P( b, x ))$
where $P$ and $Q$ are uninterpreted predicates and $x,b$ and $c$ are of some uninterpreted 
sort $S$.
The purified form $\purify{\psi}$ is computed as follows.
First, a conversion of $\psi$ to clausal normal form results in the clauses:
\[
\begin{array}{l}
F_0 := \{ \neg P( b, c ),\,
          \neg Q( b, c ) \vee \forall x\, P( b, x ),\,  
          Q( b, c ) \vee \neg \forall x\, P( b, x )
        \}
\end{array}
\]
We replace the occurrence of $\neg \forall x\, P( b, x )$ in the third clause of $F_0$
with $\neg P( b, k )$, where $k$ is a fresh variable of sort $S$,
and replace the positive occurrence in the second clause with
a fresh proxy variable $a$ of sort $\bool$.
We obtain the quantifier-free set of purified clauses $F$ and equivalences $A$:
\[
\begin{array}{l@{\hspace{1em}}l@{\hspace{1em}}l}
F & := & \{ \neg P( b, c ), \neg P( b, k ) \vee Q( b, c ), a \vee \neg Q( b, c ) \} \\
A & := & \{ a \Leftrightarrow \forall x\, P( b, x ) \}
\end{array}
\]
It is not hard to see that $\phi$ and $F \cup A$ are equisatisfiable in $T$.
\qed
\end{example}

\begin{figure}
\begin{framed}
\fmsolve$( F, A )$ \\[1ex]
{\bf Input:} A set $F$ of purified $\Sigma$-clauses and a set $A$ of equivalences \\
{\bf Output:} \sat or \unsat
\smallskip

\begin{enumerate}
\item 
Find a satisfying assignment $M$ for $F$.  
If none is found, return \unsat.
\smallskip

\item 
Construct a $\Sigma$-interpretation $\M$ that satisfies $M$.
Let $\vals$ be a minimal set of $\Sigma$-terms such that,
for each uninterpreted sort $S$ of $\Sigma$,
every element of $S^\M$ is denoted by a term in $\vals$
(that is, for all $s \in S^\M$ there is a $v \in \vals$ 
such that $s = \evaluate{v}{\M}$).
\smallskip

\item 
%Let $Q$ be the set of quantified formulas that are active in $M$.
For each $\forall \tupl{x}\,\varphi$ where $a \liff \forall \tupl{x}\, \varphi \in A$ and $a \in M$,
\smallskip

\begin{enumerate}
\item
let $\instans{\tupl{x}}$ be the set of substitutions from $\tupl{x}$ to terms in $\vals$ chosen by $\He( \M, \forall \tupl{x}\,\varphi )$;
\smallskip

\item
let $(F, A ) = (F \cup F', A \cup A')$
where 
$( F', A' ) = \purify{ \{ \neg a \vee \varphi \sigma \mid \sigma \in \instans{\tupl{x}} \}}$.
\end{enumerate}
\smallskip

If each of the sets $\instans{\tupl{x}}$ was empty, return \sat, 
otherwise go to Step 1.

%otherwise let $\purify{ \{ \neg a \vee \varphi \sigma \mid \sigma \in \instans{\tupl{x}} \} } = ( F', A' )$,
%return \fmsolve( F \cup F', A \cup A' )$.
%Add $\{ \neg a \vee \varphi \sigma \mid \sigma \in \instans{\tupl{x}} \}$ to $F$.
%If each of these sets is empty, return ``satisfiable'', otherwise go to Step 1.
\end{enumerate}
\end{framed}
\caption{Finite Model Finding Procedure \fmsolve, parameterized by a quantifier instantiation heuristic $\He$. }
\label{fig:fmf-proc}
\end{figure}

\subsection*{Model finding procedure}

Figure~\ref{fig:fmf-proc} describes a finite model finding procedure called
\fmsolve 
that takes as input a set $F$ and a set $A$, where $(F, A) = \purify\phi$ 
for some $\Sigma$-formula $\psi$, and 
tries to determine the satisfiability of $F \cup A$
by adding to $F$ instances of the quantified formulas that occur in $A$.
The procedure is parametrized by an instantiation heuristic $\He$ 
for the quantified formulas.

In Step 1, it looks for a satisfying assignment $M$ for $F$.\footnote{% 
Recall that a satisfying assignment is a $T$-satisfiable set 
of $\Sigma$-literals that propositionally entail $F$.
}
This assignment can be found using the \dpllts procedure from the previous section.
If no satisfying assignment can be found, the procedure terminates 
with \unsat, for ``unsatisfiable.''
Otherwise, in Step 2, it constructs a $\Sigma$-interpretation $\M$ 
that satisfies $M$.
% which we refer to as a \emph{candidate model}.
In doing so, however, it considers only $\Sigma$-interpretations $\M$ 
that interpret the uninterpreted sorts of $\Sigma$ as finite sets.
This makes it feasible to actually construct the set $\vals$ used in Step 3.
In that step, the procedure considers the set of quantified formulas 
that are \emph{active} in $M$, that is, 
those whose proxy variable occurs positively in $M$.
It adds new constraints $F' \cup A'$ to $F \cup A$ based on instances 
of quantified formulas chosen by the heuristic $\He$, 
which takes as input a model and a quantified formula.
%We say a quantified formula $\forall \tupl{x}\, \varphi$ is active in $M$ if its proxy variable $a$ occurs positively in $M$.
We consider only heuristics that are sound with respect to models: if $\He$ returns no instances for quantified formulas in Step 3,
it is because $\M$ satisfies all active quantified formulas in $M$, 
and so the procedure terminates with \sat, for ``satisfiable.''
%Note that purification is in turn applied to each $\neg a \vee \varphi\sigma$ added in Step 3, 
%and thus each instance may correspond to multiple clauses that are added to $F$.
%\ct{Where do we state and prove the soundness of this procedure?}
\medskip

\begin{thm}
\label{thm:fmf-correct}
For all inputs $F,A$ for \fmsolve, the following hold.
\begin{enumerate}
\item If the method for finding satisfying assignments $M$ for $F$ in Step 1 is sound, then the procedure \fmsolve returns \unsat only if $F \cup A$ is $T$-unsatisfiable.
\item If for all inputs, $\He( \M, \forall \tupl{x}\,\varphi )$ returns the empty set only if $\M \models \forall \tupl{x}\,\varphi$, then the procedure \fmsolve 
returns \sat only if $F \cup A$ is $T$-satisfiable.
\end{enumerate}
\end{thm}
\begin{proof}
To show Point~1, assume the method for finding satisfying assignments $M$ for $F$ in Step 1 is sound.
Thus, when the procedure returns \unsat, we have that $F$ is $T$-unsatisfiable.
Since the formulas added to $F$ and $A$ in Step 3 preserve satisfiability, we have that our input is $T$-unsatisfiable as well.

To show Point~2,
the procedure returns \sat when $\He( \M, \forall \tupl{x}\,\varphi )$ returns
the empty set for all quantified formulas where $a \liff \forall \tupl{x}\,\varphi \in A$ and $a \in M$.
Assume $\M \models \forall \tupl{x}\,\varphi$ for all such formulas.
Then, $F \cup A$ is satisfied by a model $\M'$ where
$a^{\M'} = (\forall \tupl{x}\,\varphi)^\M$ for each $a \liff \forall \tupl{x}\,\varphi \in A$ and $a \not\in M$,
and where all other symbols are interpreted as in $\M$.
Since during all iterations of the procedure $F \cup A$ remains a superset of 
the original input, we have that the latter is satisfied by $\M'$ as well.
\qed
\end{proof}

\ 

The following sections will examine in more detail the main ideas
behind the three steps of procedure \fmsolve.
In Section~\ref{sec:t-fcc}, we describe techniques 
for finding satisfying assignments in Step 1.
These assignments have models that interpret uninterpreted sorts as sets of minimal cardinality
and are used in Step 2.
In Section~\ref{sec:fmf-model-construct}, we describe methods for constructing such models.
Finally, in Section~\ref{sec:fmf-mbqi}, we describe quantifier instantiation heuristics 
that can be used to choose sets of substitutions for Step 3.

Although Theorem~\ref{thm:fmf-correct} holds in general for inputs that involve
several theories, for simplicity, we restrict ourselves in the following
to problems in \euf only, that is, involving only uninterpreted sorts and function symbols.
Under these restrictions,
we provide arguments for the correctness of the three steps of \fmsolve in
Theorems~\ref{thm:fixed-card-dpllts} and~\ref{thm:choose-inst-ge},
which ensure the correctness of our finite model finding procedure 
according to Theorem~\ref{thm:fmf-correct}.

%==============================================================================
\section{EUF with Finite Cardinality Constraints (\fcc)} \label{sec:t-fcc}
%==============================================================================

In this section, we introduce techniques for finding
satisfying assignments in Step 1 of procedure \fmsolve
from Figure~\ref{fig:fmf-proc}.
We will focus on satisfying assignments that have \emph{small} models, that is,
models which interpret the uninterpreted sorts of our signature as finite sets of minimal size.
To do this, we introduce an extension of the theory \euf with finite cardinality constraints (\fcc).
We describe its signature ($\Sigma_\fcc$) and semantics,
give a satisfiabiliy procedure for conjunctions of literals in this theory, 
and describe how it can be integrated into the \dpllts architecture.
%We then describe an efficient solver for \fcc that works well in practice,
%and how it can be used for minimizing the number of equivalence classes in the congruence closure maintained by the solver.
Finally, we discuss a strategy, \emph{fixed-cardinality \checkfcc},
which ensures that upper bounds are incrementally established for all uninterpreted sorts.

\begin{definition}[\fcc]
Let \euf be the theory of equality and uninterpreted functions over some signature $\Sigma_\euf$.
The theory \fcc of \emph{\euf with finite cardinality constraints} is the extension of \euf
obtained as follows.
The signature $\Sigma_\fcc$ of \fcc extends $\Sigma_\euf$ with
a constant $\card[S,k]$ of sort $\bool$
for each sort $S$ of $\Sigma_\euf$ and integer $k > 0$.
Its models are all $\Sigma_\fcc$-interpretations 
that satisfy each atomic formula $\card[S,k]$ exactly when they interpret $S$
as a set of cardinality at most $k$.
\end{definition}

%Note that the only ground atoms in \fcc
%besides those of the form $\card[S,k]$ are equalities.
As shown below, the \fcc-satisfiability of sets of $\Sigma_\fcc$-literals
is a decidable problem.
By a reduction from graph (vertex) coloring, one can show that the problem is NP-hard.
The main idea of the reduction is to represent the set of $k$ colors as a sort $\mathsf{C}$
and represent the vertices of the graph as variables of sort $\mathsf{C}$.
An edge between two vertices $x$ and $y$ is encoded as the constraint $x \tneq y$.
The cardinality constraint on $\mathsf{C}$ is encoded by $\card[\mathsf{C},k]$.
It is not difficult to see that given a model $\M$ of \fcc
(which is finitely representable), checking whether $\M$ satisfies a set of 
$\Sigma_\fcc$-literals can be done in polynomial time.
It follows that this satisfiability problem is NP-complete.

We prove its decidability by providing an effective satisfiability procedure.
The procedure relies on computing certain congruence closures of sets of constraints,
so we start by introducing that notion.

\begin{definition}[Congruence Closure]
Let $M$ be a set of literals, in any signature, and 
let $\terms_M$ be the set of all terms (and subterms) occurring in $M$.
The \emph{congruence closure $\cc M$ of $M$} is 
the smallest set of literals such that
\begin{enumerate}
\item $M \subseteq \cc M 
         \subseteq \{s \teq t \mid s,t \in \terms_M\} \cup \{s \tneq t \mid s,t \in \terms_M\}$;
\item for all $s,t \in \terms_M$, $\cc M \ent[\euf] s \teq t$ iff $s \teq t \in \cc M$. 
\end{enumerate}
\end{definition}

By construction, the relation $\{(s,t) \mid s \teq t \in \cc M\}$ induced by $\cc M$ is 
a congruence, and hence an equivalence, relation over $\terms_M$.
For brevity, we will identify $\cc M$ with its induced equivalence relation when convenient.
It can be shown (see, e.g.,~\cite{BaaNip-98}) that 
$(i)$ $\cc M$ is computable whenever $M$ is finite, 
and 
$(ii)$ if $M$ is satisfiable
it is satisfied by an interpretation $\M$ that interprets each sort $S$
as the set $\vals_S = \{v^S_1, \ldots, v^S_{n_S}\}$ consisting of
an arbitrary representative $v^S_i$ for each of the $n_S$ equivalence classes 
of $\cc M$ over terms of sort $S$.  
We call $\M$ a \emph{normal model} of $M$.
Our procedure will seek to find normal models for given input sets $M$ of literals.

%Given a normal model $\M$, 
%a \emph{model assignment} is a pair, written $t \mapsto v$, where $t$ is a term and $v$ is a value from the domain of $\M$.
%Given a congruence closure $\cc E$ for $M$ with representatives $\vals_S$ for sort $S$,
%we may construct a set of model assignments $\evmap{M}$ consisting of $t \mapsto v_i$ for all $t \in \terms_M$,
%where $v_i$ is the representative term in the equivalence class of $\cc E$ containing $t$.
%We will call $\evmap{M}$ an \emph{evaluation map} for $M$.
%We will write $\evmap{M}(t)$ to denote the value that $t$ is mapped to in $\evmap{M}$.

\subsection{Decision Procedure} 
\label{sec:fcc-basic}

This section presents a decision procedure for the satisfiability of 
sets of constraints in the theory \fcc.
For now, we limit ourselves to signatures $\Sigma_\fcc$ whose set of sorts 
consists of a single (uninterpreted) sort $\So$.
Figure~\ref{fig:dp-fcc} gives the decision procedure for the satisfiability problem
in this case.
As input, the procedure takes a set $M$ consisting of cardinality constraint literals 
for $\So$ and
equalities and disequalities over ground $\Sigma_\fcc$-terms of sort $\So$,
and terminates with \sat or \unsat.

\begin{figure}
\begin{framed}
{\bf Input:} A set $M$ of $\Sigma_\fcc$-literals \\
{\bf Output:} \sat or \unsat
\begin{enumerate}
\item 
%Compute the congruence closure $\cc M$ of $M$.
If $s \teq t \in \cc M$ for some $s \tneq t \in M \cup \{\false \tneq \true\}$, return \unsat.
\item 
If $M$ contains no positive cardinality literals, return \sat; \\
otherwise, let $k$ be the smallest integer such that $\card[\So,k] \in M$.
\item
If $\neg \card[\So,j] \in M$ for some $j \geq k$, return \unsat.
\item 
If there are $k$ or fewer equivalence classes in $\cc M$, return \sat.
\item 
\label{it:split}
If there exists two terms $s$ and $t$ in distinct equivalence classes of $\cc M$ 
such that $M \not\ent[\euf] s \tneq t$,
run the procedure recursively on $M \cup s \teq t$ and $M \cup s \tneq t$, 
returning \sat if either of the two subcalls returns \sat, and returning \unsat otherwise.
\item 
\label{it:last}
%Otherwise, there are $k+1$ equivalence classes with representatives $t_1, \ldots, t_{k+1}$
%where $M \models t_i \tneq t_j$ for all $1 \leq i \neq j \leq k+1$;
Return \unsat.
\end{enumerate}
\end{framed}
\caption{Decision Procedure for \fcc.}
\label{fig:dp-fcc}
\end{figure}

\begin{lem}
\label{lem:fcc-correct}
The procedure in Figure~\ref{fig:dp-fcc} is sound, complete and terminating 
for every set $M$ of $\Sigma_\fcc$-literals.
\end{lem}
\begin{proof}
Soundness)
Let us start by observing that splitting the problem based on equalities $s \teq t$, 
as done in Step~\ref{it:split} of the procedure,
is trivially sound since all models of \fcc satisfy exactly one 
of $s \teq t$ and $s \tneq t$.
The procedure answers \unsat in one of the following cases:
\begin{enumerate}
%$(i)$ when 
\item
an equality $s \teq t$ is entailed by $M$ where $s \tneq t$ is also in 
$M \cup \{\false \tneq \true\}$,
%$(ii)$ when 
\item
conflicting literals $\card[\So,k]$ and $\neg \card[\So,j]$ are asserted for $j \geq k$, 
or 
%$(iii)$ when 
\item
there exist $k+1$ terms (each in a different equivalence class) 
that are entailed to be mutually disequal by $M$.
\end{enumerate}

\noindent
For the first case, it is immediate that $M$ has no models.
For conflicts in the second case, no model can be constructed with 
both at most $k$ and at least $j+1$ elements in the domain of $\So$.
For conflicts in the third case, note that if the procedure reaches Step~\ref{it:last}
there must be $k+1$ equivalence classes with representatives $t_1, \ldots, t_{k+1}$, say,
where $M \ent[\euf] t_i \tneq t_j$ for all $1 \leq i < j \leq k+1$;
hence no model can be constructed satisfying $\card[\So,k]$.
\smallskip

Termination)
It is easy to see that when the procedure recurses in Step~\ref{it:split},
the set of equalities and disequalities in $M$ without the cardinality constraints
is satisfiable.
Let $C$ be a set collecting the equivalence classes of $\cc M$
and let $[t]_M$ denote the equivalence class of a term $t$.
We argue that the splitting on the equality of $s$ and $t$ done at Step~\ref{it:split}
decreases the size of the set
\begin{equation}
 E_M := \{\, ([u]_M,[v]_M) \in C \times C \,\mid\, 
             [u]_M \neq [v]_M,\: M \not\ent[\euf] u \tneq v\, \}
\end{equation}
in other words, it decreases
the number of equivalence classes that are pairwise not entailed to be disequal.
In either branch of the split on $s \teq t$,
no equivalence classes are created (although two existing ones are possibly merged), 
and $([s]_M,[t]_M)$ is no longer an element of $E_M$ in the recursive call.
When $E_M$ becomes empty, the procedure is guaranteed to terminate,
since either more than $k$ equivalence classes are entailed to be distinct, 
in which case the procedure answers \unsat, 
or there are at most $k$ equivalence classes,
in which case the procedure answers \sat. 
\smallskip

Completeness)
The procedure answers \sat 
when the congruence closure $\cc M$ contains no equality whose negation occurs in $M$, 
and either there is no positive cardinality literal in $M$,
or $\cc M$ has at most $k$ equivalence classes
where $k$ is the smallest integer such that $\card[\So,k] \in M$.
In either case, we can construct a model where $S$ is interpreted 
as a set of size $j$, with $j \leq k$ for all $\card[\So,k] \in M$
and $j \geq k$ for all $\neg \card[\So,k] \in M$.
If $j$ is greater than the number of equivalence classes in $\cc M$,
arbitrary new elements can be added to the domain of $\So$
without affecting the satisfiability of the equalities and disequalities in $M$.
\qed
\end{proof}
\medskip

An immediate consequence of this lemma is that constraint satisfiability in $\fcc$ is decidable.

\begin{proposition}
The $\fcc$-satisfiability of sets of $\Sigma_\fcc$-literals is decidable.
\end{proposition}

The completeness argument in Lemma~\ref{lem:fcc-correct}
also suggests a constructive proof of the following result.
%\ct{This is not a corollary of the decidability proof.
%This property is true on its own. Promoted to a proposition.
%}
\begin{proposition}
\label{prop:fcc-finite}
Every satisfiable set of $\Sigma_\fcc$-literals has a finite model.
\end{proposition}
%\begin{proof}
%Whenever our decision procedure answers satisfiable, we may construct a finite model 
%where $S$ is interpreted as the set of representative terms from each equivalence class,
%as well as a (finite) number of additional elements that ensure that all literals of the form $\neg \card[\So,j]$ in $M$ are satisfied.
%\qed
%\end{proof}

We point out that in the absence of cardinality constraints 
the decision procedure in Figure~\ref{fig:dp-fcc} reduces 
to the standard congruence closure procedure used to decide the satisfiability
of constraints in \euf.
SMT solvers supporting \euf have theory solvers that essentially implement 
that procedure.

\subsection{Integration into \dpllts}
\label{sec:fcc-dpllts-opt}

Our decision procedure for \fcc can be integrated into the \dpllts framework
by capitalizing on the existence of a theory solver for \euf ($T_e$).
We effectively extend such a solver modularly with facilities to reason 
about cardinality constraints as well.
Since \fcc is an extension of \euf, we now replace the latter with the former
in the framework and make $T_e = \fcc$.
Recall the strategy outlined in Figure~\ref{fig:dpllt-strategy} of Section~\ref{sec:dpllt-strat}.
In the following, we detail how
the methods $|weak\_effort|$ and $|strong\_effort|$ of this strategy are implemented for \fcc.
%Due to its use of weak effort checks, our solver will be able to eagerly report conflicts that prune the overall search space.

For simplicity, we maintain the restriction for now that \fcc contains 
a single uninterpreted sort $\So$.
%For handling more than one sort, multiple invocations of the techniques can be invoked independently. 
Also, when convenient, we identify equivalence classes of terms
with their representative terms.

\subsubsection*{Weak Effort Check}
%\label{sec:fcc-dpllts-opt-weak}

\begin{figure}[t]
\begin{program}
\PROC |weak\_effort\_\fcc|( M, F, C ) \BODY 
  \IF l_1, \ldots, l_n \ent[\euf] \bot \text{ for some } l_1, \ldots, l_n \in M
    \text{Apply \conflict{e} with } C := \compl{l}_1 \vee \cdots \vee \compl{l}_n, \text{ return } \mathsf{false}
  \ELSEIF \card[\So,k], \neg \card[\So,j] \in M \text{ for } j>k,
    \text{Apply \conflict{e} with } C := \neg \card[\So,k] \vee \card[\So,j], \text{ return } \mathsf{false}
  \ELSEIF \card[\So,k] \in M \text{ and } M \ent[\euf] \distinct( t_1, \ldots, t_{k+1} )
    \text{Apply \learn{e} to } \neg \card[\So,k] \vee \neg \distinct( t_1, \ldots, t_{k+1} ), \text{ return } \mathsf{false}
  \ELSE
    \text{return } \mathsf{true}
\ENDPROC
\end{program}
\caption{
Weak effort check for \fcc.
}
\label{fig:fcc-dpllts-opt-weak-strategy}
\end{figure}

At weak effort, we recognize conflicting states of three different forms, outlined in Figure~\ref{fig:fcc-dpllts-opt-weak-strategy}.
First, if we are unable to construct a congruence closure for $M$ that is consistent with the disequalities from $M$,
we identify a subset $\{l_1, \ldots, l_n\}$ of $M$ that is \euf-unsatisfiable
and apply \conflict{e} to it.
Second, if $M$ contains the conflicting cardinality constraints $\card[\So,k] \in M$ and
$\neg \card[\So,j]$ with $j > k$,
we construct the conflict clause $\neg \card[\So,k] \vee \card[\So,j]$.
Third, we may recognize cases when $M$ contains a literal of the form $\card[\So,k]$
while its other literals entails that
$k+1$ terms $t_1, \ldots, t_{k+1}$ are pairwise disequal.
In this case, we use \learn{e} to add the lemma 
$\neg \card[\So,k] \vee \neg \distinct( t_1, \ldots, t_{k+1} )$ to the current set $F$
of clauses,
where $\distinct( t_1, \ldots, t_{k+1} )$ is shorthand for the conjunction of disequalities stating that 
the terms $t_1, \ldots, t_{k+1}$ are pairwise distinct.\footnote{%
Note that $\neg \card[\So,k] \vee \neg \distinct( t_1, \ldots, t_{k+1} )$
is a valid formula of \fcc.
}
We will refer to a lemma of this form as a \emph{clique lemma}.
We assume that this instance of \learn{e} is applied only 
when the resultant clause does not occur in $F$.
In practice, this can be achieved either 
by maintaining a cache of learned clauses or
by ensuring \propagate{e} is applied to completion between each call.
We apply \learn{e} because the constraint 
$\neg \distinct( t_1, \ldots, t_{k+1} )$ may contain literals not belonging to $F$.
We could alternatively apply \conflict{e} to construct a conflict clause of form 
$\compl{l}_1 \vee \ldots \vee \compl{l}_n \vee \neg \card[\So,k]$,
where $\{l_1, \ldots, l_n\}$ is a subset of $M$ 
that entails $\distinct( t_1, \ldots, t_{k+1} )$.
However, we have found that in practice this is inefficient, 
as many different sets of literals can be found for essentially the same conflict.

\paragraph{Generating clique lemmas}
For the purposes of discovering and learning clique lemmas,
we incrementally construct and maintain on the side a \emph{disequality graph} $D$ for $\So$, 
whose vertices correspond to the equivalence classes of terms of sort $\So$
induced by the congruence closure of $M$,
and whose edges represent disequalities in $M$ between terms in different equivalence classes.
In this representation, a sufficient condition for discovering a conflict 
reduces to finding a $(k+1)$-clique in $D$.
Now, even just checking for the presence of a $(k+1)$-clique 
in a $n$-vertex graph is too expensive in general---as 
%its worst-case complexity is $O(n^{k+1} (k+1)^2)$.
this is an NP-complete problem~\cite{Garey:1974:SNP:800119.803884}.
For this reason, the weak effort check of our procedure uses an incomplete check 
for potential cliques.
This is done by partitioning the vertices of the graph into suitable
subsets that we call \emph{regions}.
After defining regions formally, we explain below how we exploit them 
to discover clique-related conflicts efficiently in practice.

\begin{definition}[$k$-Region]
\label{def:region}
Let $D = (V, E)$ be an undirected graph and let $R$ be a subset of $V$.
For all vertices $v \in R$, let $\mathsf{ext}(v)$ be the number of edges between $v$ and vertices not in $R$.
We say $R$ is a \emph{$k$-region} of $D$ if for all $0 < i \leq k$,
the size of the set $\{ v \mid v \in R, \mathsf{ext}(v) \ge i \}$ is smaller than $k-i$. 
A \emph{$k$-regionalization} $\mathcal{R}_D$ of $D$ is a partition of $V$ into $k$-regions.
We will refer to it simply as a \emph{regionalization} when $k$ is understood or not important.
\end{definition}

Regionalizations are useful for us because they facilitate the discovery of cliques.

\begin{lem}
If $\mathcal{R}_D$ is a $k$-regionalization of a graph $D$ and 
$D$ contains a $k$-clique $C$,
then all the vertices of $C$ reside in the same region of $\mathcal{R}_D$.
\end{lem}
\begin{proof}
If $k \leq 1$, the statement is trivial.
Otherwise, assume by contradiction $D$ contains $k$-clique $C = C_1 \cup C_2$ 
for non-empty $C_1, C_2$ where, for some region $R$ of $\mathcal{R}_D$, 
$v \in R$ for all $v \in C_1$ and $v \not\in R$ for all $v \in C_2$.
Say $\vert C_2 \vert = i$, and thus $\vert C_1 \vert = k-i$.
Since $C$ is a $k$-clique, $\mathsf{ext}(v)$ must be at least $i$ for all $v \in C_1$,
contradicting the assumption that $R$ is a region.
\qed
\end{proof}
\medskip

Notice that any graph $D = (V, E)$ has a trivial regionalization, 
with just one region which contains all vertices in $V$.

\begin{example}
Consider the constraints 
$\{ c_1 \tneq c_2, c_2 \tneq c_3, c_3 \tneq c_4 \}$,
all over sort $\So$, and the partition $\{ \{ c_1, c_2 \}, \{ c_3, c_4 \} \}$.
This partition is a 3-regionalization in the disequality graph induced by this set,
because a 3-clique can span two regions only if it contains two vertices with interregional edges,
and this partition only has one such edge.
Adding the disequality $c_2 \tneq c_4$ or $c_1 \tneq c_4$
breaks the regionalization invariant.
\qed
\end{example}

Let us examine how to maintain a $k$-regionalization in an (initially empty) \emph{evolving graph} $D$,
a data structure supporting the dynamic allocation of vertices and edges, 
as well as the merging of vertices.
In our framework, where $D$'s vertices correspond to equivalence classes of terms and 
edges to disequalities between them, 
these operations are triggered by operations performed on the data structure
that stores the congruence closure $\cc M$ of the current assignment $M$.
In particular, a vertex $v_e$ is added to $D$ when a new equivalence class $e$ is created,
which happens whenever a new term is added to $M$;
an edge between the vertices $v_1$ and $v_2$ corresponding to equivalence classes $e_1$ and $e_2$
is added to $D$ when the disequation $t_1 \tneq t_2$ is added to $M$, 
for some term $t_1$ in $e_1$ and $t_2$ in $e_2$; 
and 
two vertices are merged, in a single vertex that inherits their edges, 
when their corresponding equivalence classes are merged during the computation of $\cc M$.

\begin{figure}[t]
\begin{program}
\PROC |fix\_region|( R, \mathcal{R}_D ) \BODY 
  \IF R \text{ is not a $k$-region }
    \text{choose some } R' \in \mathcal{R}_D, \text{ where } R' \neq R
    \mathcal{R} := \mathcal{R} \setminus \{ R, R' \} \cup \{ R \cup R' \}
    |fix\_region|( \{ R \cup R' \}, \mathcal{R}_D )
  \FI
\ENDPROC
\end{program}
\vspace{-2ex}
\caption{The |fix\_region| procedure. Ensures $R \in \mathcal{R}_D$ is a $k$-region
by merging it with another $R' \in \mathcal{R}_D$, and repeating this process recursively.
}
\label{fig:fix-region}
\end{figure}

We maintain at all times a $(k+1)$-regionalization $\mathcal{R}_D$ of the graph $D$, 
where $k$ is the smallest integer such as $\card[\So,k] \in M$.\footnote{%
Recall that $\card[\So,k]$ states that sort $\So$ has at most $k$ elements.
}
As the graph $D$ is modified, it may be necessary to merge certain 
regions of the current within $\mathcal{R}_D$ 
to ensure the invariant in Definition~\ref{def:region} holds.
The procedure $|fix\_region|$ from Figure~\ref{fig:fix-region} ensures 
that a set $R$ within $\mathcal{R}_D$ is a $k$-region by merging it as needed 
with another set $R'$ in $\mathcal{R}_D$, and repeating this process 
recursively until $R$ becomes a $k$-region.
As a heuristic, we choose the $R'$ with the highest density of interregional edges to $R$.

Assuming we have a regionalization $\mathcal{R}_D$ for graph $D$,
here is how we construct a regionalization $\mathcal{R}_{D'}$ 
for graph $D'$ resulting from an addition to $M$.
In the following, $\mathcal{R}(v)$ denotes the region in a regionalization $\mathcal{R}$ 
that contains the vertex $v$.

\begin{description}
\item[Adding Vertices:]
When a vertex $v$ is added to $D$, $\mathcal{R}_{D'}$ is the result of adding the singleton region $\{ v \}$ to $\mathcal{R}_D$.
\smallskip

\item[Adding Edges:]
When we add an edge $( v_1, v_2 )$ to $D$, we have that $\mathcal{R}_{D'} = \mathcal{R}_D$ is still a partition of $V$.
However, $\mathcal{R}_D( v_1 )$ or $\mathcal{R}_D( v_2 )$ may not be regions of $D'$.
We apply the procedure $|fix\_region|$ first to $( \mathcal{R}_D( v_1 ), \mathcal{R}_{D'} )$
and then to $( \mathcal{R}_D( v_2 ), \mathcal{R}_{D'} )$ to ensure that $\mathcal{R}_{D'}$ is a regionalization.
\smallskip

\item{Merging Vertices:}
When a vertex $v_1$ is merged with another vertex $v_2$ in $D$, we have that $D'$ is a quotient graph of $D$, that is,
$D'$ contains a new vertex, call it $u$, connected to all vertices that are connected to either $v_1$ or $v_2$ in $D$.
If $\mathcal{R}_D(v_1)$ is equal to $\mathcal{R}_D(v_2)$, 
let $R$ be $(\mathcal{R}_D(v_1) \cup \{ u \}) \setminus \{ v_1, v_2 \}$.
Then $\mathcal{R}_{D'}$ is equal to $(\mathcal{R}_D \cup R) \setminus \{ \mathcal{R}_D(v_1) \} $.
To ensure $\mathcal{R}_{D'}$ is a regionalization, we apply $|fix\_region|$ to $( R, \mathcal{R}_{D'} )$.
If $\mathcal{R}_D(v_1)$ is not equal to $\mathcal{R}_D(v_2)$, 
let $\{ v_i, v_j \} = \{ v_1, v_2 \}$,
$R_i = (\mathcal{R}_D(v_i) \cup \{ u \}) \setminus \{ v_i \}$, and
$R_j = \mathcal{R}_D(v_j) \setminus \{ v_j \}$.
Then, $\mathcal{R}_{D'}$ is equal to $(\mathcal{R}_D \cup \{ R_i, R_j \}) \setminus \{ \mathcal{R}_D(v_1), \mathcal{R}_D(v_2) \}$.
We apply $|fix\_region|$ to $( R_i, \mathcal{R}_{D'} )$ and subsequently to $( R_j, \mathcal{R}_{D'} )$.
\end{description}

Additionally, when $\card[\So,k']$ is asserted for $k' < k$, we discard the $(k+1)$-regionalization and rebuild a $(k'+1)$-regionalization.

%We are interested in finding cliques of size $k+1$ in the disequality graph $D$ 
%for $S$ induced by $M$, which would indicate that $M$ is unsatisfiable.
%For this purpose, our solver maintains a $(k+1)$-regionalization $\mathcal{R}_D$ of $D$.
Given a $(k+1)$-regionalization $\mathcal{R}_D$ of $D$,
we will call each region in $\mathcal{R}_D$ with at least $k+1$ vertices a \emph{large region}, and all others \emph{small regions}.
For the purposes of efficiently discovering $(k+1)$-cliques
during weak effort checks,
we maintain a \emph{watched set} of $k+1$ vertices for each large region $R$ in $\mathcal{R}_D$, which we will write as $w(R)$.
This set is incrementally updated when vertices are added or removed from regions, and when regions are combined. 
%Maintaining watched sets of vertices helps recognize conflicting states during a weak effort check. 
If there exists a large region $R$ in $\mathcal{R}_D$ where each vertex in $w(R)$ is connected,
then we add the clique lemma $\neg \card[\So,k] \vee \neg \distinct(t_1, \ldots, t_{k+1})$ to $F$ using the rule \learn{e},
where $w(R) = \{ t_1, \ldots, t_{k+1} \}$.

\subsubsection*{Strong Effort Check}
%\label{sec:fcc-dpllts-opt-strong}

\begin{figure}[t]
\begin{program}
\PROC |strong\_effort\_\fcc|( M, F, C ) \BODY 
  \text{let $k$ be the smallest integer such that } \card[ \So, k ] \in M \\
  \text{let $t_1, \ldots, t_n$ be the equivalence class representatives of sort $\So$ in $\cc M$} \\
  \IF n>k
    \text{choose $1 \leq i \lt j \leq n$ such that } M \not\ent[\euf] t_i \tneq t_j \\
    \text{apply \learn{e} to } t_i \teq t_j \vee t_i \tneq t_j \\
    \text{return } \mathsf{false}
  \ELSE
    \text{return } \mathsf{true}
\ENDPROC
\end{program}
\caption{
Strong effort check for \fcc.
}
\label{fig:fcc-dpllts-opt-strong-strategy}
\end{figure}

Recall from Section~\ref{sec:dpllt-strat} that a strong effort check 
must determine that the current set of constraints is consistent, or 
otherwise report a conflict or lemma.
The strong effort check of the \fcc solver is given in Figure~\ref{fig:fcc-dpllts-opt-strong-strategy}.
If $\card[ \So, k ] \in M$ for some (minimal) $k$, and 
there are more than $k$ equivalence class of sort $\So$ 
in the congruence closure of $M$,
then we choose two equivalence class representatives $t_i$ and $t_j$ and 
apply \learn{e} to add the splitting lemma 
$( t_i \teq t_j \vee t_i \tneq t_j )$ to $F$.
In practice, we also insist that future applications of \decide on the atom $t_i \teq t_j$ should be invoked with positive polarity.
If the number of equivalence classes is less than or equal to $k$, then the procedure returns $\true$, indicating that $M$ is $\fcc$-satisfiable.

The choice of $t_i$ and $t_j$ is guided by the watched set of vertices within regions.
In particular,
if there is a large region $R$ in $\mathcal{R}_D$, we know that $w(R)$ does not form a clique.
We choose $t_i, t_j$ to be two vertices from $w(R)$ that are not connected in $D$.
Otherwise, if there are no large regions in $\mathcal{R}_D$ and 
there are more than $k$ vertices in $D$,
then there must be at least two small regions.
We select two regions $R_i$ and $R_j$ based on a heuristic 
(namely, the maximum density of interregional edges),
combine them into a new region $R_i \cup R_j$, apply $|fix\_region|$ to $R_i \cup R_j$, and repeat.

We illustrate the operation of the \fcc solver with a couple of examples.

\begin{example}
Consider the constraints 
$\{ a \teq f( b ),\, b \teq f( c ),\, a \tneq b,\, b \tneq c,\, \card[\So,2] \}$ 
where all terms are over the single sort $\So$.
First, the \fcc solver computes the congruence
$\{\{a,f( b )\},\, \{b, f( c )\},\, \{c\}\}$.
Using $a, b, c$ as the representatives,
the solver builds the disequality graph with edges
$\{(a,b),(b,c)\}$.
Since $\card[\So,2]$ limits the size of $\So$ to at most 2, 
the solver generates the lemma $a \teq c \lor a \tneq c$.
Adding the constraint $a \teq c$ produces no conflicts and 
allows the \fcc solver to answer ``satisfiable''.
\qed
\end{example}

\begin{example}
Consider the constraints 
$\{ c_1 \teq c,\, c_4 \teq c,\, c_1 \tneq c_2,\, c_2 \tneq c_3,\, c_3 \tneq c_4,\, \card[\So,2]\}$
with all constants of sort $\So$. 
The corresponding disequality graph for these constraints contains 
a clique of size 3.
By discovering that clique, the \fcc solver can conclude that it is impossible
to shrink the model to 2 elements, and hence reports a clique lemma of the form
$\neg \distinct( c_1, c_2, c_3 ) \lor \lnot\card[\So,2]$.
\qed
\end{example}

Because of congruence constraints, guesses on merge lemmas may sometimes lead 
to inconsistencies when constructing the congruence closure, unless we
compute and propagate all entailed disequalities---which is usually not done,
for efficiency.
This is demonstrated in the following example.

\begin{example}
Consider the constraints 
$\{c_3 \teq f(c_1),\, c_4 \teq f(c_2),\, c_3 \tneq c_4,\, \card[\So,2] \}$
where all the terms have sort $\So$.
Unless the \euf subsolver propagates the entailed literal $c_1 \tneq c_2$,
the \fcc solver will construct the disequality graph 
$(V, E) = (\{c_1, c_2, c_3, c_4 \}, \{(c_3, c_4)\})$ for $\So$.
Say we decide to apply \learn{e} on $( c_1 \teq c_2 \vee c_1 \tneq c_2 )$,
and then the literal $c_1 \teq c_2$ is added to our set of constraints.
The subset
$\{c_3 \teq f(c_1),\,$ $c_4 \teq f(c_2),\, c_3 \tneq c_4,\,$ $c_1 \teq c_2\}$
will then be found unsatisfiable by congruence closure.
In contrast, adding the equalities $c_1 \teq c_3$ and $c_2 \teq c_4$ to our set will produce
a model of the required cardinality.
\qed
\end{example}

\begin{comment}
It is immediate that the solver in this section
described in this section is also sound.
To argue that it is terminating, notice that $|fix\_region|$ is terminating 
since each recursive call to this procedure
reduces the number of regions in $\mathcal{R}_D$ by one,
and similarly for repeated calls of the strong effort check.
Additionally, we have that all introduced literals
(either those when reporting clique lemmas at weak effort, or when splitting on equalities at strong effort)
are taken from the finite set of equalities and disequalities between terms occurring in our original clause set $F_0$.
Since the conditions for answering satisfiable are the same as those as mentioned in Lemma~\ref{lem:fcc-correct}
and the solver mentioned in this section is terminating, it is also complete.
\end{comment}

We now state the correctness of our \fcc procedure as integrated in the \dpllts framework.
In the following, we let \checkfcc denote the strategy 
that applies the rules of \dplltsfcconly 
according to Figure~\ref{fig:dpllt-strategy},
with the weak and strong effort checks described in this section.

\begin{thm}
\label{thm:fcc-dpllts-correct}
\checkfcc is a sound, complete and terminating strategy 
for every set of ground clauses $F_0$.
\end{thm}
\begin{proof}
%\ct{this is too sketch. The connection to soundness, completeness and termination should be made explicit}
% AJR : this is taken care of by proposition 1. The purpose of that proposition is to make theorems like this easier.
Notice that the weak and strong effort methods in this section legally apply \dpllts rules, that is,
they apply \conflict{e} only to clauses whose negated literals imply a contradiction
and \learn{e} to clauses that hold in all models.
We follow the three requirements for weak and strong effort checks as described in Proposition~\ref{prop:dpllt-strat}.

To show the first point, the only literals introduced by applications of \learn{e} (call them $L_\fcc$)
are equalities and disequalities between terms occurring in $F_0$.
Clearly $L_\fcc$ is finite.
To show the second point, the weak and strong effort methods in this section $\mathsf{false}$ only when they apply at least one rule.
To show the third point, $|strong\_effort\_\fcc|( M, F, C )$ returns $\mathsf{true}$ only
when the congruence closure of $M$ contains $k$ or fewer equivalence classes for all $\card[\So,k] \in M$.
In such states, we are guaranteed that $M$ is satisfiable in \fcc.
\qed
\end{proof}

\subsection{Establishing Finite Cardinalities} 
\label{sec:fcc-min-card}

We have now shown that a theory solver for \fcc can be integrated into the \dpllts architecture
with support for eager conflict detection through the use of weak effort checks.
In this section, we show an approach that makes use of this solver 
for answering the following problem:
\emph{given an input $F$, find the smallest integer $n > 0$ 
such that $F \wedge \card[\So,n]$ is satisfiable}.

A straightforward scheme for solving this problem is the following.
First, use the solver to determine if $F \wedge \card[\So,1]$ is satisfiable, and answer satisfiable if so.
If this is unsatisfiable, use the solver to determine if $F \wedge \card[\So,2]$ is satisfiable, and so on.
Due to Proposition~\ref{prop:fcc-finite}, this process is guaranteed to terminate when $F$ is satisfiable.
A clear disadvantage of this scheme is that, 
in the absence of conflict analysis, 
it diverges when $F$ is unsatisfiable.
%The following approach does not have this limitation.
This section describes an alternative approach 
that overcomes this limitation.
At a high level, our approach modifies the weak effort check 
of the \fcc solver by introducing splits on cardinality constraints
$(\card[\So,k] \vee \neg \card[\So,k])$,
and deciding upon literals of the $\card[\So,k]$ for the minimal feasible $k$.
%In detail,
%if the weak effort check from Figure~\ref{fig:fcc-dpllts-opt-weak-strategy} does not result 
%in the application of \learn{i} or \conflict{i},
%we find the least integer $k \geq 1$ such that $\neg \card[\So,k]$ is not in $M$.
%If $\card[\So,k]$ is not a literal in $F$, then we add the splitting lemma 
%$(\card[\So,k] \vee \neg \card[\So,k])$ to $F$.
%Otherwise, we apply \decide on $\card[\So,k]$.
%In other words, we decide the minimal possible cardinality bound for $\So$ given the current context.
Before formally defining this approach,
we discuss a generalization that is applicable to signatures with multiple uninterpreted sorts.
%In this approach, all other applications of \decide occur when $\card[\So,k] \in M$ for some $k$.

\begin{figure}[t]
\begin{minipage}[t]{.4\linewidth}
\begin{program}
\PROC |weak\_effort\_fc\_\fcc|( M, F, C ) \BODY   
  \text{Let $k$ be the least $\mathbb{N}$ s.t. } k \geq n \text{ and } \neg \card[\Sigma,k] \not\in M
  \IF |fix|( \card[\Sigma,k], M, F, C ) = \mathsf{false}
    \text{return } \mathsf{false}
  \FI 
  \text{For each } \So_i \in \Sigma, \text{let $k_i$ be the least $\mathbb{N}$ s.t. } \neg \card[\So_i,k_i] \not\in M 
  \IF |fix|( \card[\So_i,k_i], M, F, C ) = \mathsf{false} \text{ for a minimal $i$ }
    \text{return } \mathsf{false}
  \FI 
  \IF k_1 + \ldots + k_n > k
    \text{Apply \conflict{e} to } C := ( \displaystyle\vee_{i=1}^n \card[\So_i, k_i-1] \vee \neg \card[\Sigma, k] ) \\
    \text{return } \mathsf{false}
  \FI
  \text{return } |weak\_effort\_\fcc|( M, F, C )
\ENDPROC
\end{program}
\end{minipage}
\begin{minipage}[t]{.4\linewidth}
\begin{program}
\PROC |fix|( a, M, F, C ) \BODY 
    \IF a \not\in \lits{F}
      \text{Apply \learn{e} to } (a \vee \neg a) \\
      \text{return } \mathsf{false}
    \ELSEIF a \not\in M
      \text{Apply \decide to } a \\
      \text{return } \mathsf{false}
    \ELSE
      \text{return } \mathsf{true}
\end{program}
\end{minipage}

\caption{A version of the weak effort check procedure of the \fcc solver
that fixes the cardinality of uninterpreted sorts $\{ \So_1, \ldots, \So_n \}$ in signature $\Sigma$
according to a fair strategy.
}
\label{fig:weak-effort-fix-card-ms}
\end{figure}

\subsubsection*{Extension to Multiple Sorts}

Consider the case when our signature $\Sigma$ contains multiple sorts $\So_1, \ldots, \So_n$.
Given a set of input clauses $F$, we wish to determine that either $F$ is unsatisfiable, 
or find a tuple $( k_1, \ldots, k_n )$ such that $F \wedge \card[\So_1, k_1] \wedge \ldots \wedge \card[\So_n, k_n]$ is satisfiable.
To find such a tuple, a challenge is to devise a strategy that is \emph{fair}.
As an illustrative example, consider the formula $( c \tneq d \vee \varphi )$, 
where $c$ and $d$ are constants of sort $\So_1$, and the formula $\varphi$ does not have a model where $\So_2$ is interpreted as a finite set.\footnote{%
Observe that $\varphi$ must contain universal quantifiers for this to be the case.
}
Clearly this formula has a model where the cardinality of sorts $\So_1$ and $\So_2$ are $2$ and $1$ respectively.
However, in the absence of a fair strategy, 
a naive approach could search for models of size $( 1, 1 )$, $( 1, 2 )$, $( 1, 3 )$, and so on, ad infinitum.

To devise a strategy for finite model finding that is fair in the presence of multiple sorts,
we extend the signature $\Sigma$ of \fcc to include \emph{signature cardinality constraints} 
$\card[\Sigma, k]$, constants of sort $\bool$ for each integer $k > 0$.
Let $\Sigma$ be a signature containing uninterpreted sorts $\So_1, \ldots, \So_n$.
Let $\mathcal{I}$ be a $\Sigma$-interpretation that interprets sort $\So_i \in \Sigma$ as a set of size $k_i$ for $1 \leq i \leq n$.
Then, $\mathcal{I}$ satisfies $\card[\Sigma, k]$ if and only if $k_1 + \ldots + k_n \leq k$.

Figure~\ref{fig:weak-effort-fix-card-ms} gives an extension of the weak effort check of the \fcc solver
that introduces cardinality constraints for the purposes of finding small models.
In detail, we first find the minimal natural number $k$ such that the literal $\neg \card[\Sigma, k]$ does not occur in $M$.
Using the sub-routine |fix|, if the atom $\card[\Sigma, k]$ does not occur in $F$,
we apply \learn{e} to add $( \card[\Sigma, k] \vee \neg \card[\Sigma, k] )$ to $F$.
If it does occur in $F$, we apply \decide to $\card[\Sigma, k]$.
We then do the same for each of the uninterpreted sorts $\So_1, \ldots, \So_n$ in our signature.
If these steps do not apply a rule,
then $M$ contains the literals $\card[ \Sigma, k ]$ 
and $\neg \card[ \So_i, \ell ]$ for each $1 \leq \ell < k_i$, $i = 1, \ldots, n$.
We then check if $\card[ \Sigma, k ]$ is in conflict with the negatively asserted cardinality constraints.
In particular, $k_1 + \ldots + k_n > k$, we return a conflict of the form 
$( \card[\So_1, k_1-1] \vee \ldots \vee \card[\So_n, k_n-1] \vee \neg \card[\Sigma, k] )$,
where we write $\card[\So_i, k_i-1]$ to denote a cardinality constraint when $k_i>1$ and $\bot$ if $k_i=1$.
Otherwise, we apply the original weak effort check of the \fcc solver from Figure~\ref{fig:fcc-dpllts-opt-weak-strategy}.

Let \emph{fixed-cardinality \checkfcc}
be the strategy that applies the rules of \dplltsfcconly according to Figure~\ref{fig:dpllt-strategy}
where the weak effort check is the one from Figure~\ref{fig:weak-effort-fix-card-ms},
and the strong effort check is the one from Figure~\ref{fig:fcc-dpllts-opt-strong-strategy}.
This strategy maintains the following invariant.

\begin{proposition}
\label{prop:fixed-card-dpllts}
Given a signature $\Sigma$ containing uninterpreted sorts $\So_1, \ldots, \So_n$,
for each execution of fixed-cardinality \checkfcc ending in $\state{M,F,C}$,
either $M$ contains no decision points, or $M$ is of the form 
$N \ \bullet$ $\card[\Sigma,k] \  M_0\ ({} \bullet \card[\So_1, k_1] M_1)$ $\cdots$ $\ ({} \bullet \card[\So_m, k_m] \  M_m) \ N'$,
for some $m$, $0 \leq m \leq n$, where $N$,$M_0$, $\ldots$, $M_m$ contain no decision points, 
$N'$ contains no decision points if $m \lt n$,
$\neg \card[\Sigma, j] \prec_M \card[\Sigma, k]$ for each $n \leq j \lt k$, 
and $\neg \card[\So_i, j] \prec_M \card[\So_i, k_i]$ for each $1 \leq i \leq m$, $1 \leq j \lt k_i$.
\end{proposition}

In other words, using the strategy fixed-cardinality \checkfcc,
minimal positive cardinality literals are the first decision literals in satisfying assignments.
This invariant follows directly from definition of the method given in Figure~\ref{fig:weak-effort-fix-card-ms}.

\begin{thm}
\label{thm:fixed-card-dpllts}
Fixed-cardinality \checkfcc is a sound, complete and terminating strategy 
for every set of ground clauses $F$.
\end{thm}
\begin{proof}
Assume our signature $\Sigma$ contains uninterpreted sorts $\So_1, \ldots, \So_n$.
Note that the weak effort check method in Figure~\ref{fig:weak-effort-fix-card-ms} 
extends our original weak effort check while additionally applying only legal
applications of \dpllts rules, noting we apply \learn{e} to tautologies of the form $(a \vee \neg a)$,
\conflict{e} to sets of literals that are collectively inconsistent according to our extension of \fcc, and
\decide to literals whose atom does not occur in $M$.
To show this strategy is sound, complete, and terminating, 
we again follow the three requirements for weak and strong effort checks as given in Proposition~\ref{prop:dpllt-strat}.

To show the first point,
we must show that the set $L_\fcc$ of literals introduced by applications of \learn{e} is finite.
For each $1 \leq i \leq n$, let $k_i$ be smallest integer greater than the number of terms of sort $\So_i$ in $F$,
and such that the literal $\neg \card[\So_i,k_i]$ does not occur in $F$.
Let $k$ be the smallest integer greater than $k_1 + \ldots + k_n$,
and such that the literal $\neg \card[\Sigma,k]$ does not occur in $F$.
We claim that the set of literals introduced by applications of \learn{e}, call them $L^{fc}_\fcc$,
are a subset of the set of all equalities and disequalities between terms from $F$,
the literals of the form $(\neg) \card[\So_i,j]$ where $1 \leq j \lt k_i$ for each sort $\So_i$,
and the literals $(\neg) \card[\Sigma,j]$ where $n \leq j \lt k$.
First,
only equalities and disequalities between terms from $F$ are introduced by applications of \learn{e}
for the same reason as in Theorem~\ref{thm:fcc-dpllts-correct}.
Second, assume by contradiction that a literal $(\neg) \card[\So_i,k_i]$ is introduced by an application of \learn{e}.
Then, it must be the case that an execution of fixed-cardinality \checkfcc results in a state where
$\neg \card[\So_i,k_i-1] \in M$.
For this to be the case, 
$\neg \card[\So_i,k_i-1]$ must be added to $M$ by \propagate{e} or \backjump.
In either case, there must exist a set of literals $l_1, \ldots, l_n$ from $F$
such that $l_1, \ldots, l_n \models_\fcc \neg \card[\So_i,k_i-1]$.
By our selection of $k_i$ this is a contradiction
since there must be at least $k_i$ terms of sort $\So_i$ in $F$ for this to be the case.
Third, for similar reasons, by our selection of $k$, 
a literal $(\neg) \card[\Sigma,k]$ cannot be introduced by an application of \learn{e} 
unless there exists an execution of fixed-cardinality \checkfcc resulting in a state where
$\neg \card[\So_i,j] \in M$ for some $\So_i$ where $j \geq k_i$.
This cannot be the case for the reasons mentioned above.

To show the second point, weak effort check method in Figure~\ref{fig:weak-effort-fix-card-ms} 
returns $\mathsf{false}$ only when it applies a rule.

To show the third point,
the strong effort check of fixed-cardinality \checkfcc is the same as the strong effort check in \checkfcc
and thus this holds for the same reason as in the proof of Theorem~\ref{thm:fcc-dpllts-correct}.
\qed
\end{proof}

\ 

Combining the results of Theorem~\ref{thm:fixed-card-dpllts} and Proposition~\ref{prop:fixed-card-dpllts},
given as input a set of ground $\Sigma$-clauses $F$,
fixed-cardinality \checkfcc will terminate either in 
$(i)$ a fail state, establishing that $F$ is unsatisfiable, or
$(ii)$ a state $\state{M,F,\none}$ where $M$ contains $\card[ \So_i, k_i ]$ for each uninterpreted sort $\So_i$ in $\Sigma$,
establishing that $F$ is satisfied by a (finite) model.

For the remainder of the paper, we assume that
Step 1 of our finite model finding procedure in Figure~\ref{fig:fmf-proc}
uses fixed-cardinality \checkfcc for finding satisfying assignments $M$.

\section{Constructing Candidate Models}
\label{sec:fmf-model-construct}

%We have seen how satisfying assignments are constructed for a ground set of clauses $F$,
%where the interpretation of uninterpreted sorts can be minimized using a theory solver for cardinality constraints, 
%and the strategy fixed-cardinality \checkfcc.
We now focus our attention to Step 2 of procedure \fmsolve from Figure~\ref{fig:fmf-proc}, 
which attempts to constructs models $\M$ of satisfying assignments $M$ 
for the input clause set $F$.
We refer to $\M$ as a \emph{candidate model}.
Note that the assignment $M$ computed by the procedure may contain occurrences
of proxy variables $a$ for quantified formulas $\forall \vec x\, \varphi$
with the variables in $\vec x$ ranging over uninterpreted or \emph{finite} sorts.
Recall that those formulas are stored in the input set $A$ 
in equivalences of the form $a \Leftrightarrow \forall \vec x\, \varphi$.
The goal of the procedure is to construct $\M$ so that it satisfies 
not just $M$ but also all its active quantified formulas
(those whose proxy variable $a$ occurs positively in $M$).
The reason is that such a model witnesses the $T$-satisfiability of
$F \cup A$.%AJR: reworded to be more accurate

To discuss the model construction we focus
on the variables and the uninterpreted sorts and function symbols of $\Sigma$,
since the interpretation of the other sorts and function symbols is fixed by the theory.
%%In fact, it is enough to consider only the uninterpreted sorts and symbols occurring in  
%%$F \cup A$.
%With no loss of generality, we assume that the set $A$ contains 
%all the uninterpreted function symbols of $\Sigma$ and contains a term of sort $S$
%for each uninterpreted sort $S$.
%Recall\ct{? was this mentioned before?} that we require that each input or output sort of the uninterpreted function symbols 
%be either uninterpreted or interpreted in $T$ as a finite set.
We construct a candidate model $\M$ by associating each uninterpreted sort $S$
with a finite set $\vals_S$ of \emph{domain elements} (i.e., $S^\M = \vals_S$).
Contrary to other model finding approaches, which use fresh symbols as domain elements,
we use the equivalence classes of $\cc M$ or, rather,
representative terms for these classes.
All interpreted sorts are interpreted in $\M$ as usual.
We extend $\cc M$ to another $T$-satisfiable set, call it $\ccc M$,
such that the representative of each equivalence class 
of interpreted sort $S_i$ in $\ccc M$ is a value from $S_i^\M$.
Such an extension is always possible since $M$ is $T$-satisfiable.

We then associate each uninterpreted function $f$ of sort 
$S_1 \times \ldots \times S_n \rightarrow S$
to a function $f^\M$ from $S_1^\M \times \cdots \times S_n^\M$ to $S^\M$.
We construct this function based on the literals in $M$ that contain $f$.
For instance, if $M$ contains $f( c ) \teq b$, then $f^\M$ is defined so that 
it maps the interpretation of $c$ to the interpretation of $b$.
Using those equalities typically produces only a partial definition for $f$.
To complete it, one can use arbitrary output values for the missing input tuples.
We describe choices for doing so in the following.

Concretely, we represent candidate $\Sigma$-models with the following data structure.

\begin{definition}[Defining map]
\label{def:defining-map}
Let $f : S_1 \times \cdots \times S_n \rightarrow S$ be an uninterpreted function symbol
of $\Sigma$ and
let $y_1, \ldots, y_n$ be distinct fresh variables of respective sort $S_1, \ldots, S_n$.
A \emph{defining map for $f$} is a finite set $\Delta_f$ 
of well-sorted (directed) equations of the form $f(t_1,\ldots,t_n) \teq v$ 
with $v \in S^\M$ and 
$t_i \in \{y_i\} \cup S_i^\M$ for $i=1,\ldots,n$,
satisfying the following requirements.
\begin{enumerate}
\item  \label{it:uniqueness}
If $s_1 \teq v_1, s_2 \teq v_2 \in \Delta_f$ with $s_1 \neq s_2$
and
$s_1$ and $s_2$ have an mgu $\sigma$, then 
\smallskip

\begin{enumerate}
\item
$\sigma$ is non-empty, and 
\smallskip

\item
$s_1\sigma \teq v \in \Delta_f$ for some $v$.
\end{enumerate}
\smallskip

\item \label{it:existence}
$f(y_1, \ldots, y_n) \teq v \in \Delta_f$ for some $v$.
\end{enumerate}
A \emph{$\Sigma$-map} is 
a set $\Delta = \bigcup_{f \in \sfuns{\Sigma}} \Delta_f$
where each $\Delta_f$ is a defining map for $f$.
\qed
\end{definition}

For the rest of this section, we will use letters $y, y_1, y_2, \ldots$
to denote variables and $c, c_1, c_2, \ldots$ to denote constant symbols.

%\ct{An example of a $\Sigma$-map would be good at this point.}
\begin{example}
The set $\{ f( c_1, y_2 ) \teq c_2, f( y_1, c_2 ) \teq c_1, f( c_1, c_2 ) \teq c_3, f( y_1, y_2 ) \teq c_3 \}$
is a defining map for $f$.
Notice that $f( c_1, y_2 )$ and $f( y_1, c_2 )$ have mgu $\{ y_1 \mapsto c_1, y_2 \mapsto c_2 \}$.
As required in point~\ref{it:uniqueness}, this mgu is non-empty and an equality of the form $f( c_1, c_2 ) \teq v$ also occurs in this set.
\qed
\end{example}

\noindent
By construction of $\Delta$, every \emph{flat term},
a $\Sigma$-term $t = f(v_1,\ldots,v_n)$
has exactly one \emph{most specific generalization} $s$ among the left-hand sides
of the equalities in $\Delta_f$,
where $s$ is a generalization of $t$ if $t = s\sigma$ for some substitution $\sigma$,
and $s$ is more specific than $s'$ if $s'$ is a generalization of $s$.
The existence of this generalization is guaranteed by Point~\ref{it:existence} 
in the definition above;
its uniqueness by Point~\ref{it:uniqueness}.
The \emph{value of $t$ in $\Delta$} is 
the value $v$ in the (unique) equality $s \teq v \in \Delta_f$.
Thus, a $\Sigma$-map $\Delta$ represents a normal model $\M$
where each uninterpreted sort $S$ is interpreted as the term set $\vals_S$ and 
each uninterpreted function symbol $f : S_1 \times \cdots \times S_n \rightarrow S$ 
is interpreted as the function $f^\M$ mapping 
every $(v_1,\ldots,v_n) \in S_1^\M \times \cdots \times S_n^\M$
to the value of $f(v_1,\ldots,v_n)$ in $\Delta$.\footnote{%
More precisely,
a $\Sigma$-map represents a family of normal models which differ only 
over the variables and the interpreted symbols of $\Sigma$. 
}

\begin{figure}[t]
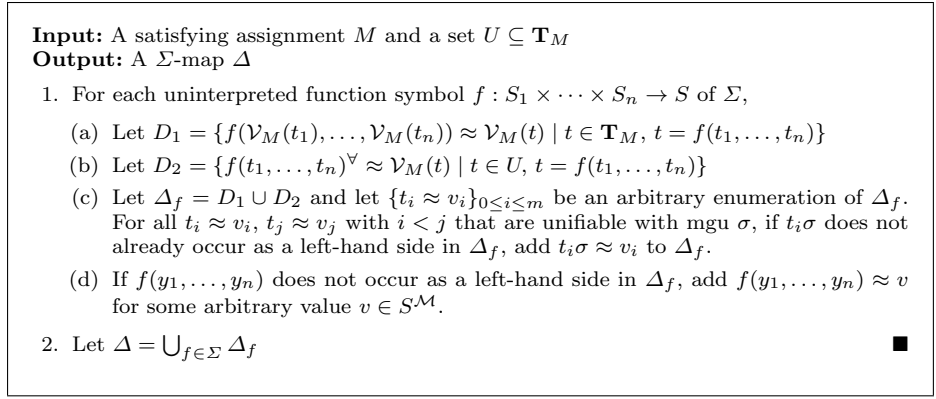

\begin{framed}
{\bf Input:} A satisfying assignment $M$ and a set $U \subseteq \terms_M$ \\
{\bf Output:} A $\Sigma$-map $\Delta$
\begin{enumerate}
%\item
%Select of a subset $U$ of $\terms_{M}$.
%\smallskip

\item
For each uninterpreted function symbol $f : S_1 \times \cdots \times S_n \rightarrow S$ 
of $\Sigma$,
\smallskip

\begin{enumerate}
\item 
Let 
$D_1 = \{ f(\evmap{M}(t_1),\ldots,\evmap{M}(t_n)) \teq \evmap{M}(t) \mid 
          t \in \terms_M,\, t = f(t_1,\ldots,t_n) \}
$
\smallskip

\item
Let 
$D_2 = \{ f(t_1,\ldots,t_n)^\forall \teq \evmap{M}(t) \mid 
          t \in U,\, t = f(t_1,\ldots,t_n) \}
$
\smallskip

\item
Let $\Delta_f = D_1 \cup D_2$ and
let $\{t_i \teq v_i\}_{0\leq i \leq m}$ be an arbitrary enumeration of $\Delta_f$.
For all $t_i \teq v_i$, $t_j \teq v_j$ with $i < j$ that are unifiable with mgu $\sigma$,
if $t_i\sigma$ does not already occur as a left-hand side in $\Delta_f$,
add $t_i\sigma \teq v_i$ to $\Delta_f$.
\smallskip

\item
If $f(y_1,\ldots,y_n)$ does not occur as a left-hand side in $\Delta_f$,
add $f(y_1,\ldots,y_n) \teq v$ for some arbitrary value $v \in S^\M$.
\end{enumerate}
\smallskip

\item
Let $\Delta = \bigcup_{f \in \Sigma} \Delta_f$
\qed
\end{enumerate}
\end{framed}
\caption{Model construction procedure.}
\label{fig:mcp}
\end{figure}

\subsubsection*{Model Construction Procedure}
%\label{sec:model-cons}
We now describe a procedure for constructing $\Sigma$-maps
from satisfying assignments.
In particular, we describe a parametrized method for completing
the partial definitions (of uninterpreted functions) induced by 
an assignment $M$.

Let $M$ be an assignment.
Recall that if $M$ is $T$-satisfiable,
it is satisfied by a normal model, that is, a model that interprets each uninterpreted sort $S$ 
as the set $\vals_S$ consisting of a representative term for each equivalence class of 
(the extension of) $M$'s congruence closure $\ccc M$.
For each term $t$, 
we write $\evmap{M}(t)$ to denote the representative of
$t$'s equivalence class in $\ccc M$.

For every uninterpreted function symbol 
$f : S_1 \times \cdots \times S_n \rightarrow S$ in $\Sigma$,
we fix $n$ distinct fresh variables $y_1, \ldots, y_n$ of respective sort 
$S_1 \ldots, S_n$.
To each uninterpreted sort $S$, we associate a distinguished ground $\Sigma$-term $e^S$,
which we will write as $e$ when $S$ is understood.
This ground term guides the selection of default values of the interpretation of 
uninterpreted function symbols in our model construction procedure, 
based on the following operation.
For a ground $\Sigma$-term $f(t_1, \ldots, t_n)$, 
we denote by $f( t_1, \ldots, t_n )^\forall$ 
the term $f( u_1, \ldots, u_n )$
where $u_i = y_i$ if $t_i = e$, and 
$u_i = \evmap{M}(t_i)$ otherwise, for $i=1, \ldots, n$.
\medskip

The non-deterministic procedure described in Figure~\ref{fig:mcp}
constructs a $\Sigma$-map from $M$ and a subset $U$ of the set of terms $\terms_{M}$ occurring in $M$.
The subset $U$ determines which terms will be used as the basis for default values
of function interpretations. 
For example, let $\M$ be a normal model induced by the defining map
constructed by the procedure in Figure~\ref{fig:mcp} for some $U$. 
If $f( e, e ) \in U$, then
the default value of $f$ in $\M$ is the value of $f( e, e )$ in $\M$.
In our implementation of the procedure, we choose the set $U$ to be
the entire set $\terms_{M}$, although other choices are possible.

\begin{example}
Consider an assignment $M$ with the following constraints 
\[
\{c_1 \teq f(c_2,e),\, c_3 \teq f(c_4,c_6),\, c_3 \teq f(e,c_4),\, c_6 \teq f( c_2, c_5 ),\, c_2 \teq c_5,\, c_4 \teq f(e,e) \}
\]
where all terms are of uninterpreted sort $\So$.
The equivalence classes in the congruence closure of $M$ are
\[
 \{ c_1, f( c_2, e ) \},\; 
 \{ c_2, c_5 \},\;  
 \{ c_3, f( c_4, c_6 ), f( e, c_4 ) \},\;
 \{ c_4, f( e, e ) \}, \{ c_6, f( c_2, c_5 ) \}, \{ e \} \ .
\]
Let $\vals_S = \{ c_1, c_2, c_3, c_4, c_6, e \}$ 
and $U = \{ f( c_2, e ), f( e, c_4 ), f( e, e ) \}$.
Following the procedure to construct the defining map $\Delta_f$,
we let:
\[\begin{array}{r@{\hspace{1em}}c@{\hspace{1em}}l}
D_1 & = & \{ f( c_2, e ) \teq c_1, f( c_4, c_6 ) \teq c_3, f( c_2, c_2 ) \teq c_6, f( e, e ) \teq c_4 \} 
\\[1ex]
D_2 & = & \{ f( c_2, y_2 ) \teq c_1, f( y_1, c_4 ) \teq c_3, f( y_1, y_2 ) \teq c_4 \}
\\[1ex]
\Delta_f & = & D_1 \cup D_2 
\end{array}\]
Since $f( c_2, y_2 )$ and $f( y_1, c_4 )$ are unifiable with $\sigma = \{ y_1 \mapsto c_2, y_2 \mapsto c_4 \}$,
and $f( c_2, c_4 )$ is not in $\Delta_f$, we add the equality $f( c_2, c_4 ) \teq c_1$ (alternatively, $f( c_2, c_4 ) \teq c_3$) to $\Delta_f$.
Finally, since $f( y_1, y_2 )$ is already in $\Delta_f$, this gives us the set
\[\begin{array}{r@{\hspace{1em}}c@{\hspace{1em}}ll}
\Delta_f & = & \{ & f( c_2, e ) \teq c_1, f( c_4, c_6 ) \teq c_3, f( c_2, c_2 ) \teq c_6, f( e, e ) \teq c_4, \\
 & & & f( c_2, y_2 ) \teq c_1, f( y_1, c_4 ) \teq c_3, f( y_1, y_2 ) \teq c_4, f( c_2, c_4 ) \teq c_1 \}
\end{array}\]
which is a complete definition for $f$.
Notice that a different selection of $U$ would have led to a different construction for $\Delta_f$.
Let $\M$ be the normal model induced by a $\Delta$ containing $\Delta_f$.
We have that, for instance, $\evaluate{f( c_2, c_3 )}{\M} = c_1$
since $f( c_2, y_2 ) \teq c_1 \in \Delta_f$ and
$f( c_2, y_2 )$ is the most specific generalization of $f( c_2, c_3 )$ among the left-hand sides of equalities in $\Delta_f$.
Similarly, we have that $\evaluate{f( c_6, c_4 )}{\M} = c_3$ and $\evaluate{f( c_3, c_3 )}{\M} = c_4$.
\qed
\end{example}

\begin{proposition}
Let $M$ be a $T$-satisfiable assignment containing only uninterpreted function symbols 
over uninterpreted sorts.
The set $\Delta$ constructed by the procedure in Figure~\ref{fig:mcp} is a $\Sigma$-map.
Moreover, the normal model $\M$ represented by $\Delta$ satisfies $M$.
\end{proposition}

\begin{proof}
To show that $\Delta$ is a $\Sigma$-map, 
we show that $\Sigma_f$ is a defining map for each function symbol $f$ of $\Sigma$.
Step 1(c) of the procedure ensures that 
Point~\ref{it:uniqueness}(b) of Definition~\ref{def:defining-map} is met for all pairs 
of equalities in $\Delta_f$,
while Step 1(d) makes sure that Point~\ref{it:existence} is met.
We prove by contradiction that Point~\ref{it:uniqueness}(a) of 
Definition~\ref{def:defining-map} also holds for $\Delta_f$.
Assume that $t \teq v_1, t \teq v_2 \in \Delta_f$ with $v_1 \neq v_2$. 
Due to our construction, both $t \teq v_1$ and $t \teq v_2$ are in $D_1 \cup D_2$.
Thus, there must exist terms 
$t = f( t_1, \ldots, t_n )$ and $s = f( s_1, \ldots, s_n )$ in $\terms_M$ 
such that $\evmap{M}(t_1) = \evmap{M}(s_1)$, $\ldots$, $\evmap{M}(t_n) = \evmap{M}( s_n )$
and $\evmap{M}(t) = v_1 \neq v_2 = \evmap{M}(s)$,
contradicting our assumption that $M$ is a (consistent) satisfying assignment.
Thus, $\Delta_f$ is a defining map for all $f \in \Sigma$, and thus $\Delta$ is a $\Sigma$-map.

For each term $f( t_1, \ldots, t_n ) \in \terms_M$, 
we have that $f( \evmap{ M }( t_1 ), \ldots, \evmap{ M }( t_n ) ) \teq \evmap{M}(t) \in D_1$,
and thus $\M( t ) = \evmap{M}( t )$.
Thus, $\M$ satisfies all equalities between pairs of terms in the same equivalence class of $\ccc M$.
Since $\ccc M$ is $T$-satisfiable, we have that $\M$ satisfies all disequalities in $\ccc M$ as well.
Since $\ccc M$ is a superset of $M$, we have that $\M$ satisfies $M$.
\qed
\end{proof}

\section{Model-Based Quantifier Instantiation}
\label{sec:fmf-mbqi}

We now focus our attention on Step 3 of our finite model finding procedure \fmsolve
from Figure~\ref{fig:fmf-proc}.
In this step, the procedure \fmsolve considers quantified formulas in the set:
\begin{equation} \label{eq:q}
\{ \forall \tupl{x}\,\varphi \mid (a \Leftrightarrow \forall \tupl{x}\,\varphi) \in A \text{ and } a \in M\}
\end{equation}
Call this set $Q$.
For each formula $\forall \tupl{x}\,\varphi \in Q$, it uses a quantifier instantiation heuristic $\He$
that returns a set of substitutions from $\tupl{x}$ to terms in the set $\vals$ constructed in Step 2.
A trivial way to implement $\He$ is to choose all such possible substitutions.
If $\tupl{x}$ is a tuple of $n$ variables each ranging of a sort with $k$ domain elements, 
this heuristics will return $k^n$ substitutions, 
which is clearly unfeasible unless both $k$ and $n$ are rather small.
Significantly more scalable heuristics can be adopted 
if it is possible to identify sets of substitutions $\sigma$ yielding 
instances $\varphi\sigma$ that are already satisfied by the current candidate model, 
as these substitutions can be safely ignored.
These heuristics are collectively known as \emph{model-based quantifier instantiation}.

A way to perform model-based quantifier instantiation, 
as implemented in the SMT solver Z3~\cite{GeDeM-CAV-09},
is to use the SMT solver itself as an oracle:
a separate copy of the SMT solver is run on another query to determine 
whether a candidate model $\M$ satisfies each quantified formula.
If it does not, a single instance that is falsified by $\M$ is added
to the current clause set $F$.
This approach incurs the performance overhead of constructing the corresponding query 
as well as initializing the oracle.
Our version of model-based instantiation relies instead upon specialized data structures 
when checking candidate models and choosing instantiations,
and may add more than one instantiation per invocation.

We describe below a model-based quantifier instantiation method 
that identifies entire sets of instances as satisfiable in $\M$
without actually generating and checking those instances individually~\cite{ReyEtAl-2-RR-13}.
The main idea is to determine the satisfiability
in $\M$ of some instance $\varphi\sigma$ of a quantified formula 
$\forall \tupl{x}\, \varphi \in Q$,
generalize $\varphi\sigma$ to a set $J$ of instances 
equisatisfiable with $\varphi\sigma$ in $\M$, and then 
look for further instances only outside that set.
The set $J$ is computed by identifying which variables of $\varphi$ actually
matter in determining the satisfiability of $\varphi\sigma$.
Technically, 
for each $\psi = \forall \tupl{x}\, \varphi$,
substitution $\sigma = \{\tupl{x} \mapsto \tupl{v}\}$ into $\vals$, and 
instance $\varphi' = \varphi\sigma$ of $\psi$, 
if $\M \models \varphi'$
we compute a partition of $\tupl{x}$ into $\tupl{x}_1$ and $\tupl{x}_2$ and
a corresponding partition of $\tupl{v}$ into $\tupl{v}_1$ and $\tupl{v}_2$
such that
$\M \models \forall \tupl{x}_2\, \varphi\{\tupl{x}_1 \mapsto \tupl{v}_1\}$;
similarly, if $\M \not\models \lnot\varphi'$ we compute a partition 
such that
$\M \not\models \forall \tupl{x}_2\, \lnot\varphi\{\tupl{x}_1 \mapsto \tupl{v}_1\}$.
In either case, we then know that 
all instances of $\varphi\{\tupl{x}_1 \mapsto \tupl{v}_1\}$ over $\vals$
are equisatisfiable with $\varphi'$ in $\M$, and so 
it is enough to consider just $\varphi'$ in lieu of all them.
We will refer to the elements of $\tupl{x}_1$ above 
as a set of \emph{critical variables for $\varphi$ (under $\sigma$)}---although
strictly speaking this is a misnomer as we do not insist that 
$\tupl{x}_1$ be minimal.

\begin{figure}[t]
\begin{program}
\PROC |eval|( \M, t, \sigma ) \BODY 
  \MATCH t \WITH
  \mid f( t_1, \ldots, t_n ) \ \to \qtab  \quad   \LOOPFOR j = 1, \ldots, n 
      \text{let } ( v_j, X_j ) = |eval|(\M, t_j, \sigma )
    \ENDLOOP
    \text{choose a critical argument subset $C$ of } \{ 1, \ldots, n \}
    \RETURN ( f^{\M}( v_1, \ldots, v_n ),\, \bigcup_{i \in C} X_i ) \
   \untab
  \mid x \ \to \ \ \RETURN ( \sigma(x),\, \{ x \} )
  \ENDMATCH
\ENDPROC
\end{program}
\vspace{-2ex}
\caption[The |eval| procedure]{The |eval| procedure for candidate model $\M$. 
Returns a pair $( v, S )$ where $( t \sigma )^\M = v$,
and $S$ is a subset of the domain of $\sigma$ that was used to compute this interpretation.
}
\label{fig:evaluate}
\end{figure}

%------------------------------------------------------------------------------
\subsection{Generalizing Evaluations}
%------------------------------------------------------------------------------

We have developed a general procedure that,
given the $\Sigma$-map of a candidate model $\M$, a term $t$, and 
a substitution $\sigma$ over $t$'s variables,
computes and returns both the value of $t\sigma$ in $\M$ and 
a set of critical variables for $\sigma$.
This procedure effectively extends to quantifier-free formulas as well
by treating them as Boolean terms---which evaluate to either $\true$ or $\false$ 
in a $\Sigma$-interpretation depending
on whether they are satisfied by the model or not.

The procedure, called |eval|, is defined recursively over its input term and 
is sketched in Figure~\ref{fig:evaluate}. 
For uniformity, we assume that function symbols and logical operators
are all in prefix form.

When evaluating a non-variable term $f(t_1, \ldots, t_n)$, 
|eval| determines a \emph{critical argument subset} $C$ for it.
This is a subset of $\{1,\ldots,n\}$ such that
the term $f(s_1, \ldots, s_n)$ denotes a constant function in $\M$
where each $s_i$ is the value computed by |eval| for $t_i$ if $i \in C$, and 
is a unique variable otherwise.
If $f$ is a logical symbol, the choice of $C$ is dictated 
by the symbol's semantics.
For instance, for ${\teq}(t_1, t_2)$, $C$ is $\{1,2\}$;
for $\lor(t_1, \ldots, t_n)$,
%\footnote{%
%We treat $\lor$ as a multi-arity symbol.
%}
it is $\{1,\ldots,n\}$ 
if the disjunction evaluates to $\false$;
otherwise, we it chooses $\{i\}$ for some $i$ where $t_i$ 
evaluates to $\true$.
If $f$ is a function symbol of $\Sigma$, 
|eval| computes $C$ by first constructing a custom index data structure 
for interpreting applications of $f$ to values.
The key feature of this data structure is that it uses information 
on the sets $X_1, \ldots X_n$ to choose an evaluation order 
for the arguments of $f$.
For example, given the term $t = f( g( x, y, z ), v_2, h( x ))$,
say that |eval| computes the values $v_1, v_2, v_3$ and the critical variable sets
$\{ x,y,z \}$, $\emptyset$, $\{x\}$ for the three arguments of $f$, respectively.
With those sets, it will use the evaluation order $(2, 3, 1)$ 
for those arguments---meaning that 
the second argument is evaluated first, then the third, etc.
Using the index data structure, it will first determine 
if $f( x_1, v_2, x_3 )$ has a constant interpretation in $\M$ for all $x_1, x_3$.
If so, then the evaluation of $t$ depends on none of its variables,
and the returned set of critical variables for $t$ will be $\emptyset$.
Otherwise, if $f( x_1, v_2, v_3 )$ has a constant interpretation in $M$, 
then the evaluation of $t$ depends on $\{ x \}$, or else
it depends on the entire variable set $\{ x, y, z \}$.

The next example gives more details on the whole process of using |eval|
to generalize a ground instance to a set of ground instances equisatisfiable 
with it in a given model.

\begin{example}
\label{ex:qi-ge}
Let $Q = \{\forall x_1\, x_2\, f( x_1 ) \teq g( x_2, b ) \lor  h( x_2, x_1 ) \tneq b\}$,
where all terms are of some sort $\So$.
Consider a candidate model $\M$ induced by a $\Sigma$-map containing the following definitions : 
\[
\begin{array}{l}
\Delta_g = \{ g(a,a) \teq c,\, g(y_1,b) \teq a,\, g(y_1,y_2) \teq b \}
 \\[.5ex]
\Delta_f  = \{ f(b) \teq b,\, f(y_1) \teq a \}
 \\[.5ex]
\Delta_h  = \{ h(y_1,y_2) \teq b \}
\end{array}
\]
Suppose $\vals_{S} = \{ a, b, c \}$.
The table below shows the bottom-up calculation performed by |eval| 
on the formula $\varphi = f(x_1) \teq g(x_2,b) \lor h(x_2,x_1) \tneq b$
with $\M$ above and $\sigma = \{ x_1 \mapsto a, x_2 \mapsto a \}$.

\begin{center}
%\small
%\scriptsize
\begin{tabular}{|@{\ }c@{\ }|@{\ }l@{\ }|@{\;}c|}
\hline
\bf input & \bf output & \bf critical arg. subset \ 
\\
\hline
$x_1$ & $(a, \{ x_1 \})$ & \ 
\\
$x_2$ & $(a, \{ x_2 \})$ & \ 
\\
$b$ & $(b, \emptyset)$ & $\emptyset$
\\
%$\true$ & $(\true, \emptyset)$ & //
%\\
$f(x_1)$ & $(a, \{ x_1 \})$ & $\{ 1 \}$
\\
$g(x_2,b)$ & $(a, \emptyset)$ & $\{ 2 \}$
\\
%\hline
%\end{tabular}
%\begin{tabular}{|@{\ }c@{\ }|@{\ }l@{\ }|@{\;}c|}
%\hline
%\bf input & \bf output & \bf critical arg. subset \ 
%\\
%\hline
$h(x_2,x_1)$ & $(b, \emptyset)$ & $\emptyset$
\\
$f(x_1) \teq g(x_2,b)$ & $(\true, \{ x_1 \})$ & $\{ 1, 2 \}$
\\
$h(x_2,x_1) \tneq b$ & $(\false, \emptyset)$ & $\{ 1, 2 \}$
\\
$f(x_1) \teq g(x_2,b) \lor h(x_2,x_1) \tneq b$ & $(\true, \{ x_1 \} )$  & $\{ 1 \}$
\\
& &
\\
\hline
\end{tabular}
\end{center}
For most entries in the table the evaluation is straightforward.
For a more interesting case, consider the evaluation of $g(x_2,b)$.
First, the arguments of $g$ are evaluated,
respectively to $(a, \{x_2\})$ and $(b, \emptyset)$.
Using an indexing data structure built from $\Delta_g$
for the evaluation order $(2, 1)$,
we determine that $g(x_2, b)$ has constant value $a$ for all $x_2$.
Hence we return an empty set of critical variables for $g(x_2,b)$.

Similarly, the fact that |eval| returns $(\true, \{ x_1 \} )$ 
for the original input formula $\varphi$ and 
the substitution $\sigma = \{ x_1 \mapsto a,\, x_2 \mapsto a \}$
means that we were able to determine that 
all ground instances of 
$\varphi\{x_1 \mapsto a \} = (f(a) \teq g(x_2,b) \lor h(x_2,a) \tneq b)$,
not just the instance $\varphi\sigma$, are satisfied in $\M$.
We can then use this information in \fmsolve to completely avoid 
generating and checking those instances.
\qed
\end{example}

%------------------------------------------------------------------------------
\subsection{A Model-Based Instantiation Heuristic}
%------------------------------------------------------------------------------
For any given quantified formula $\psi$,
the |eval| procedure allows us to identify a set of instances over $\vals$
that can be represented by a single one, 
as far as satisfiability in the candidate model $\M$ is concerned.
In this subsection, we present a quantifier instantiation heuristic
that generates a set $I$ of instances
that together represent \emph{all} instances of $\psi$ over $\vals$
that are falsified by $\M$.
This kind of exhaustiveness is crucial because it allows us to conclude
that $\M \models \psi$ by just checking that $I$ is empty.

The heuristic is implemented by a procedure that relies on |eval| for computing 
the set $I$ above, or rather, a set of substitutions for generating the elements 
of $I$ from $\psi$.
The procedure is fairly unsophisticated and quite conservative 
in its choice of representative instances,
which makes it very simple to implement and prove correct.
Its main shortcoming is that it does not take full advantage of the information 
provided by |eval|, and 
so may end up producing more representative instances than needed in many cases.

Let $\psi = \forall \tupl{x}\, \varphi \in Q$
with $\tupl{x} = (x_1, \ldots, x_n)$,
where $Q$ is the set defined in~(\ref{eq:q}).
For $i=1,\ldots,n$, let $S_i$ be the sort of $x_i$ and let
$\vals_{\tupl{x}} = \vals_{S_1} \times \cdots \times \vals_{S_n}$.
For each $S \in \{S_1, \ldots, S_n\}$,
let $<_S$ be an arbitrary total ordering over the values $\vals_S$ of sort $S$.
Let $<$ be the \emph{lexicographic}
%\footnote{%
%This is defined similarly to the standard lexicographic extension except 
%that the last component of a tuple is the most significant one,
%then the last but one, and so on.
%}
extension of these orderings to the tuples 
in $\vals_{\tupl{x}}$ and observe that $\vals_{\tupl{x}}$ is totally ordered by $<$.
We write ${\tupl{v}}_{min}$ to denote the minimum of $\vals_{\tupl{x}}$ with respect to this ordering.

For every $\tupl{v} = (v_1, \ldots, v_n) \in \vals_{\tupl{x}}$
Let $\next{i}(\tupl{v})$ denote the smallest tuple $\tupl{u}$ with respect to $<$ such that 
$\tupl{v}(j) <_{S_j} \tupl{u}(j)$ for some $1 \leq j \leq n+1-i$,
if such a tuple exists, and denote ${\tupl{v}}_{min}$ otherwise
(including when $i > n$).\footnote{%
Where $\tupl{v}(j)$ and $\tupl{u}(j)$ are the $j$-th component of $\tupl v$ and $\tupl u$,
respectively.
}
For instance, with $n = 3$, $S_1 = S_2 = S_3$ and 
$\vals_{S_1} = \{a,b\}$ with $a <_{S_1} b$, we have that
$\next{1}(a,a,a) = (a,a,b)$, $\next{2}(a,a,a) = (a,b,a)$, 
$\next{2}(a,b,a) = (b,a,a)$, $\next{3}(a,a,a) = (b,a,a)$, 
and $\next{2}(b,b,a) = {\tupl{v}}_{min} = (a,a,a)$.
Note that except in the case that $\next{i}(\tupl{v})$ is ${\tupl{v}}_{min}$, we have that $\tupl{v} \lt \next{i}(\tupl{v})$.

\begin{figure}[t]
\begin{program}
\PROC \He_{m}(\M, \forall \tupl{x}\, \varphi) \BODY
  \instans{\tupl{x}} := \emptyset; \
  t := {\tupl{v}}_{min}
  \LOOPDO
    ( v, \{x_{i_1}, \ldots, x_{i_m}\} ) := |eval|(\M, \varphi, \{ \tupl{x} \mapsto t \} )
    \IF v = \false \THEN \
      \instans{\tupl{x}} := \instans{\tupl{x}} \cup \{ \{ \tupl{x} \mapsto t \} \}
    \FI
    t := \next{i}(t) \text{ where } i = n+1 - \maxf \{ 0, i_1, \ldots, i_m \}
  \ENDLOOPDO t \neq {\tupl{v}}_{min}
  \RETURN \instans{\tupl{x}}
\ENDPROC
\end{program}
\vspace{-2ex}

\caption{
A model-based instantiation heuristic $\He_{m}$, where $\tupl{x} = (x_1, \ldots, x_n)$.
}
\label{fig:choose}
\end{figure}

Our instantiation heuristic $\He_{m}$ is given in Figure~\ref{fig:choose}.
It takes in a quantifier-free formula $\varphi$ 
with variables $\tupl{x}$ and 
returns a set $\instans{\tupl{x}}$ of substitutions $\sigma$ for $\tupl{x}$
such that $\M \not\models \varphi\sigma$.
At each execution of its loop
the procedure implicitly determines with |eval| 
a set $I$ of instances of $\varphi$ 
that are equisatisfiable with $\varphi\{\tupl{x} \mapsto \tupl{v}\}$ in $\M$,
where $\tupl{v}$ is the tuple stored in the program variable $t$.
The next value $t_{next}$ for $t$ is a greater tuple chosen to maintain
the invariant that all the tuples between $t$ and $t_{next}$ generate
instances of $\varphi$ that are in $I$.
To see that, it suffices to observe that 
these tuples differ from $t$ only in positions
that correspond to non-critical variables of $\varphi$,
namely those before position $i$ where $x_i$ is the first critical variable 
of $\varphi$ in the enumeration $x_1, \ldots, x_n$.
This observation is the main argument in the proof of the following result.

\begin{lem}
Let $\tupl{v}_0, \ldots, \tupl{v}_m$ be all values successively taken
by variable $t$ at the beginning of the loop in $\He_{m}$.
Let $v_{max}$ be the maximum element of $\vals_{\tupl{x}}$.
Then for all $j=1, \ldots, m$,
\begin{enumerate}
\item
$\tupl{v}_{j-1} < \tupl{v}_j$,
\item
for all $\tupl{u}$ with $\tupl{v}_{j-1} \leq \tupl{u} < \tupl{v}_j$, \ 
$\M \models \varphi\{\tupl{x} \mapsto \tupl{u}\}$ \ iff \ 
$\M \models \varphi\{\tupl{x} \mapsto \tupl{v}_{j-1}\}$,

\item
for all $\tupl{u}$ with $\tupl{v}_m \leq \tupl{u} \leq \tupl{v}_{max}$, \ 
$\M \models \varphi\{\tupl{x} \mapsto \tupl{u}\}$ \ iff \ 
$\M \models \varphi\{\tupl{x} \mapsto \tupl{v}_m\}$.
\end{enumerate}
\end{lem}
\begin{proof}
(Sketch)
The first statement is immediate since for all $j = 1 \ldots m$, we have $v_j = \next{k}(\tupl{v}_{j-1})$ for some $k$ and $v_j \neq \tupl{v}_{min}$.
To show the second statement for a $j$,
assume $\tupl{v}_j = \next{k}(\tupl{v}_{j-1})$ for some $k$.
For each $\tupl{u}$ where $\tupl{v}_{j-1} \leq \tupl{u} < \tupl{v}_j$,
we have that $\tupl{u}(\ell) = \tupl{v}_{j-1}(\ell)$ for all $\ell \geq k$.
For all $\ell \lt k$, the $|eval|$ procedure determined that the variable $x_\ell$ was not a critical variable for $\varphi$.
Since $\tupl{u}$ and $\tupl{v}_{i-1}$ vary on only these variables,
we have $\M \models \varphi\{\tupl{x} \mapsto \tupl{u}\}$ iff $\M \models \varphi\{\tupl{x} \mapsto \tupl{v}_{j-1}\}$.
The third statement holds for similar reasons as the second.
\qed
\end{proof}

\begin{thm}
\label{thm:choose-inst-ge}
The set $\instans{\tupl{x}}$ returned 
by $\He_{m}(\M, \varphi, \tupl{x})$ is empty if and only if $\M \models \forall \tupl{x}\, \varphi$.
\end{thm}
\begin{proof}
Thanks to the previous lemma, 
if there is an instance of $\varphi$ that is falsified by $\M$,
then $\He_{m}$ will consider at least one $\tupl{v}_i$ for which $\varphi \{ \tupl{x} \mapsto \tupl{v}_i \}$ evaluates to $\false$,
and hence it will return at least one instance.
Conversely, if all instances of $\varphi$ are satisfied by $\M$, 
then all instances of $\varphi$ considered by $\He_{m}$ evaluate to $\true$,
and hence it will return no instances.
\qed
\end{proof}

We remark that, for the model finding purposes of procedure \fmsolve,
there is no need for the procedure $\He_{m}$ to compute 
the full set $\instans{\tupl{x}}$ once it contains at least one substitution.
Any non-empty subset would suffice to trigger a (more incremental) 
revision of the current candidate model $\M$.
That said, our current implementation does compute the whole set and adds all 
the corresponding instances to the clause set $F$ before computing another model for it.
Our experiments show that computing and using one substitution at a time
is worse for overall performance than computing and using
the full set $\instans{\tupl{x}}$.

\begin{example}
Consider the quantified formula $\forall x_1\, x_2\, \varphi$ and candidate model $\M$ from Example~\ref{ex:qi-ge}.
Assume that $a <_S b <_S c$.
The result of running $\He_m$ on $\M$, $\varphi$ and $\tupl{x} = ( x_1, x_2 )$ is summarized in the table below.
Each row in the column shows the value of variable $t$ at the beginning of the loop in $\He_m$,
the result of computing |eval|, the substitution (if any) added to $\instans{\tupl{x}}$ on that iteration,
and the computation of the next tuple of terms $\next{i}(t)$.
\begin{center}
\begin{tabular}{|@{\ }c@{\ }|@{\ }c@{\ }|@{\ }c@{\ }|@{\ }c@{\ }|@{\ }c@{\ }|@{\;}c|}
\hline
Iteration & $t$ & |eval|$( \M, \varphi, \{ \tupl{x} \mapsto t \})$ & Add to $\instans{\tupl{x}}$ & $i$ & $\next{i}(t)$
\\
\hline
1 & $(a,a)$ & $(\true, \{ x_1 \} )$ & $\emptyset$ & 2 & $(b,a)$ \\
2 & $(b,a)$ & $(\false, \{ x_1 \} )$ & $\{ x_1 \mapsto b, x_2 \mapsto a \}$ & 2 & $(c,a)$ \\
3 & $(c,a)$ & $(\true, \{ x_1 \} )$ & $\emptyset$ & 2 & $(a,a)$ \\
\hline
\end{tabular}
\end{center}
We begin by setting $t$ to ${\tupl{v}}_{min} = (a,a)$.
As demonstrated in Example~\ref{ex:qi-ge}, 
we have that |eval|$( \M, \varphi, \{ x_1 \mapsto a, x_2 \mapsto a \})$ returns the pair $(\true, \{ x_1 \} )$.
The first component of this pair 
indicates that $\evaluate{\varphi \{ x_1 \mapsto a, x_2 \mapsto a \}}{\M} = \true$,
and hence we do not add this substitution to $\instans{\tupl{x}}$.
The second component of this pair indicates moreover that
this interpretation did not depend on the value of $x_2$, and hence
$(\varphi \{ x_1 \mapsto a,\, x_2 \mapsto v \})^\M = \true$ for all values of $v$.
Thus, we need not consider $t = (a,b)$ or $t = (a,c)$.
Instead, on the second iteration, we consider $\next{2}(a,a) = (b, a)$.
Subsequently, 
|eval|$( \M, \varphi, \{ x_1 \mapsto b, x_2 \mapsto a \})$ returns the pair $(\false, \{ x_1 \} )$.
This indicates that 
$(\varphi \{ x_1 \mapsto b, x_2 \mapsto v \})^\M = \false$ for all values of $v$.
We add the substitution $\{ x_1 \mapsto b, x_2 \mapsto a \}$ to $\instans{\tupl{x}}$ only.
Finally, on the third iteration, 
|eval|$( \M, \varphi, \{ x_1 \mapsto c, x_2 \mapsto a \})$ returns the pair $(\true, \{ x_1 \} )$;
we add no substitutions to $\instans{\tupl{x}}$, and the loop terminates.
Overall, $\He_m$ returns the singleton set of substitutions $\{ \{ x_1 \mapsto b, x_2 \mapsto a \} \}$.
\qed
\end{example}

\subsection{Enhancement: Heuristic Instantiation}
\label{sec:heur-inst}

Modern SMT solvers rely on syntatic heuristic instantiation methods 
for finding unsatisfiable instances for quantified formulas~\cite{Detlefs03simplify:a,DBLP:conf/cade/MouraB07,ReynoldsTinelliMoura14}.
In these methods, quantified formulas are instantiated based on pattern matching.
For instance, the solver may choose to instantiate the quantified formula $\forall x\, P( f( x ) ) \Rightarrow Q( x )$
based on the substitution $\{ x \mapsto c \}$ when
$P( f( c ) )$ is a ground term occurring in its current satisfying assignment.
This technique is often referred to as \emph{E-matching}.
We found that heuristic instantiation-based E-matching can be helpful in the context of our finite model finding approach as well,
because the instances it generates are helpful in quickly ruling out candidate models
that are obviously spurious.

A quantifier instantiation heuristic $\He$, such as the model-based one from
the previous section, can be enhanced
by applying heuristic instantiation with a higher priority.
That is, we may consider a modified quantifier instantiation heuristic that 
first computes the set of instances $\instans{\tupl{x}}$ returned by
E-matching for a quantified formula $\forall \tupl{x}\,\varphi$.
If this set is non-empty, it returns $\instans{\tupl{x}}$; otherwise it
returns the instances from the original heuristic $\He$ on $\forall \tupl{x}\,\varphi$.

In practice, we have found that it is best to apply
heuristic quantifier instantiation \emph{after} finding 
a satisfying assignment with a bounded number of equivalence classes.
By waiting to apply quantifier instantiation until after a satisfying assignment of this form can be constructed, 
we can avoid pitfalls common to $E$-matching.
In particular, having a finite cardinality for uninterpreted sorts ensures
that only a finite number of terms are unique up to congruence, 
thus ensuring that $E$-matching, which is non-terminating in general,
will eventually return instances that rule out the current upper bound on cardinality,
or terminate with no instances.
We discuss the impact of heuristic instantiation further in Section~\ref{sec:results-fmf}.

\section{Results}
\label{sec:results}

We implemented all features mentioned in this paper inside \cvc~\cite{CVC4-CAV-11},
a state-of-the-art SMT solver based on the \dpllts architecture.
This section presents experimental results on this implementation.\footnote{%
Details can be found at {\url{http://cs.uiowa.edu/~ajreynol/TPLP-fmf}}.
}
We separate this section into two sets of experiments, 
the first to evaluate the relative effectiveness of various strategies 
for the \fcc solver, and 
the second to evaluate the model finder's overall performance  
when used with quantified formulas. 
For the second set of experiments, we compare our model finder against 
state-of-the-art SMT solvers and automated theorem provers.

%------------------------------------------------------------------------------
\subsection{\fcc Solver Evaluation}
%------------------------------------------------------------------------------

We first examine the effectiveness of approach to handling ground 
problems in the theory of \euf with finite cardinality constraints (\fcc).
In this section, all experiments were run on a Linux machine with an 8-core 2.60GHz 
Intel\textsuperscript{\textregistered} Xeon\textsuperscript{\textregistered} E5-2670 processor with 16GB of RAM.

We tested various configurations of the \fcc solver, 
starting with the default configuration {\bf cvc4+f}, 
which contains the region-based enhancements described 
in Section~\ref{sec:fcc-dpllts-opt},
where conflicting states are reported by using clique lemmas of the form 
$(\lnot \distinct( t_1, \ldots, t_{k+1} ) \lor \lnot \card[S,k])$.
We also tested a configuration, {\bf cvc4+fe}, which reports conflict clauses of the form $( \compl{l}_1 \vee \ldots \vee \compl{l}_n \vee \neg \card[S,k] )$,
where $l_1, \ldots, l_n$ are equalities and disequalities that entail $\distinct( t_1, \ldots, t_{k+1} )$.
This configuration avoids the introduction of new equalities into the search
(contained in the expansion of $\distinct$), 
but has the disadvantage that it can generate different conflict clauses
for essentially the same clique.
Additionally, we considered configuration {\bf cvc4+f-r}, 
which differs from {\bf cvc4+f} only in that regionalizations have always 
just one region per sort $S$, encompassing the entire disequality graph for $S$.

We also evaluated the MACE-style approach to finite model finding described in related work~\cite{mccune-1994},
which we implemented in the configuration {\bf cvc4+mace}.
In the case of a set of ground clauses $F$ involving a single sort, 
if $\terms_F$ is the set of all terms in $F$ and 
$c_1, \ldots, c_k$ are fresh constants serving as domain constants,
this configuration checks the satisfiability of 
\begin{equation}
\label{eq:mace}
 F \land \distinct(c_1, \ldots, c_k) \land 
 \displaystyle\bigwedge\limits_{t \in \terms_F} 
   (t \teq c_1 \lor \ldots \lor t \teq c_k )
\end{equation}
for $k = 1, 2, \ldots$ until (\ref{eq:mace}) is found satisfiable for some $k$.
Then, the minimal model size for $F$ is $k$.
A major and well-known shortcoming of this approach is 
the introduction of unwanted symmetries in the problem due to the use of domain constants.
\cvc can address this issue to some extent since it incorporates 
symmetry breaking techniques directly at the ground \euf level~\cite{DehEtAl-CADE-11}.

\begin{figure}[t]
\centering
%\begin{tabular}{ll}
\includegraphics[scale=.20, angle=270]{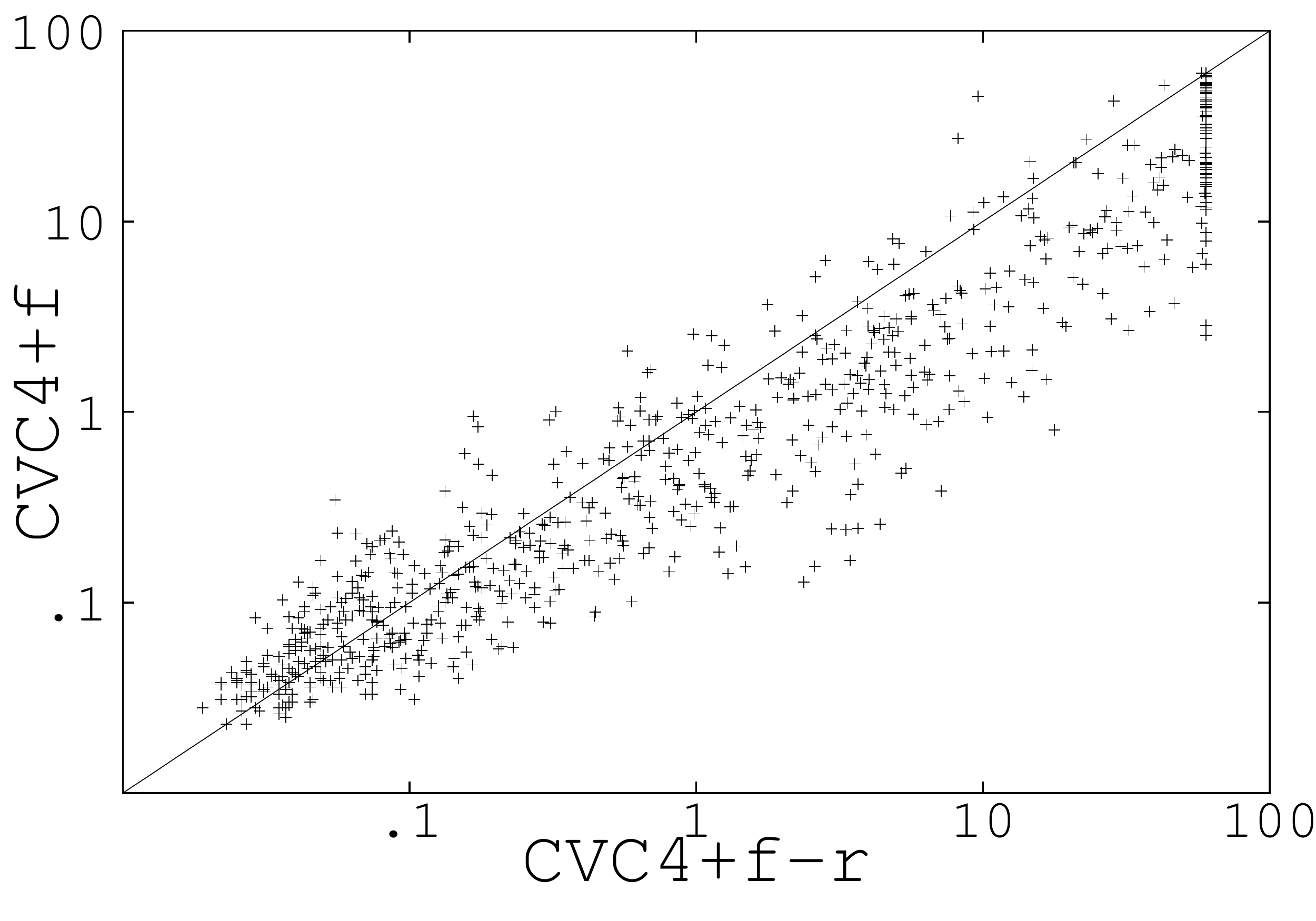} \ \  
\includegraphics[scale=.20, angle=270]{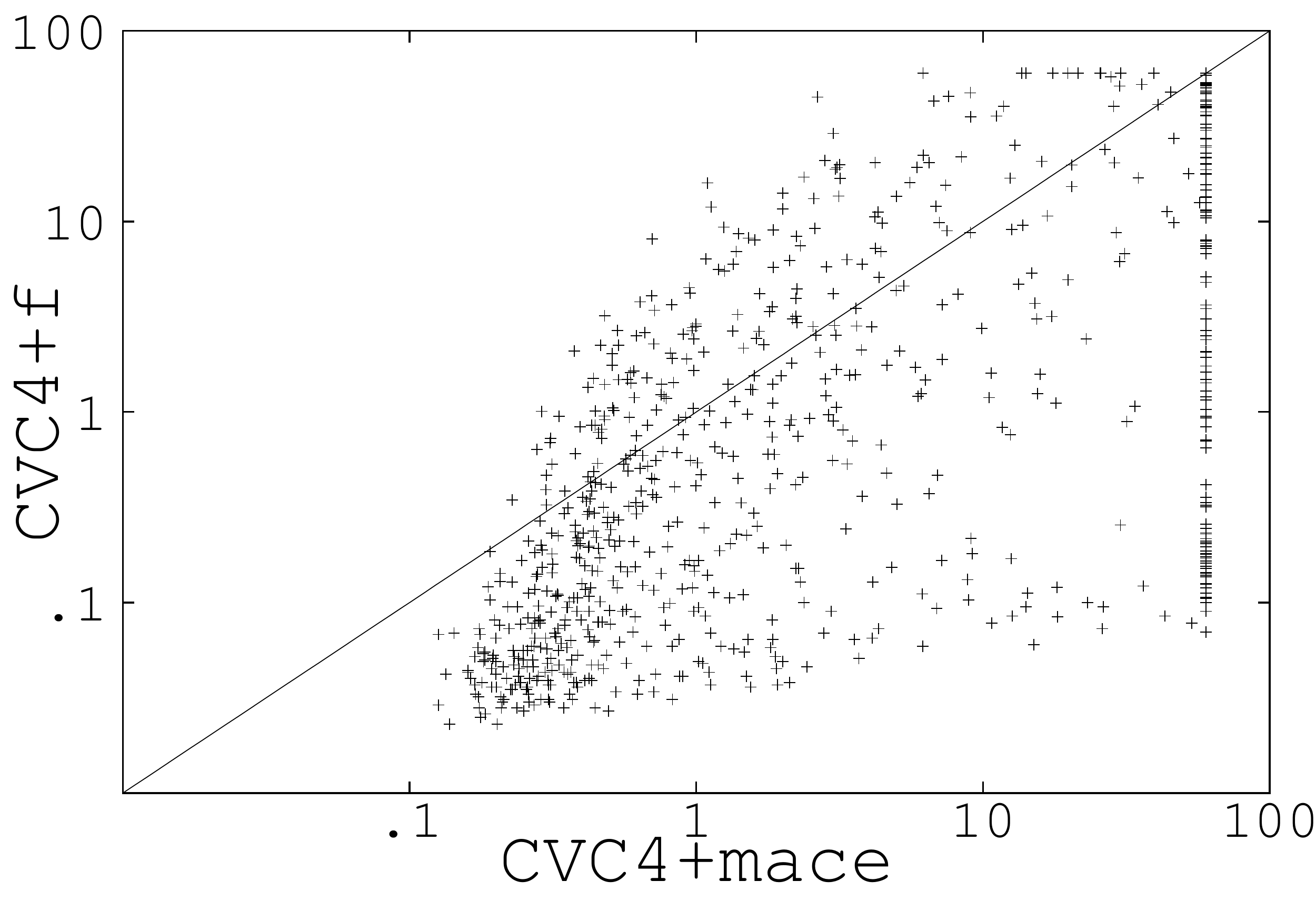}
%\end{tabular}
\caption[Results for randomly generated benchmarks]{Results for randomly generated benchmarks.
Runtimes are on a log-log scale.
}
\label{fig:results-random}
\end{figure}

We considered satisfiable benchmarks encoding randomly generated 
graph coloring problems and consisting of a conjunction of disequalities 
between constants of a single sort.
In particular, we considered a total of 793 non-trivial problems containing 
between 20 and 50 unique constants and between 100 and 900 disequalities, and 
measured the time it takes each configuration to find a model of minimum size,
with a 60 second timeout.
For the benchmarks we tested, 
the configuration {\bf cvc4+f} solves the most benchmarks 
within the time limit: 723.
The configuration {\bf cvc4+f} was an order of magnitude faster than {\bf cvc4+fe} on most benchmarks,
with the latter only being able to solve 309 benchmarks within the time limit.
This strongly suggests that generating explanations for cliques 
in conflict lemmas involving cardinality constraints is not
an effective approach in this scheme.

Figure~\ref{fig:results-random} compares the performance of the configuration {\bf cvc4+f} against {\bf cvc4+f-r} and {\bf cvc4+mace}.
The second scatter plot clearly shows that 
the {\bf cvc4+f} configuration generally requires less time and 
solves more benchmarks (723 vs. 664) 
than {\bf cvc4+f-r}, confirming the usefulness 
of a region-based approach for clique detection.
The third scatter plot compares {\bf cvc4+f} against {\bf cvc4+mace}.
The latter configuration was able to solve only 617 benchmarks and
generally performed poorly on benchmarks with larger model size.
The median model size of the 123 benchmarks solved only 
by {\bf cvc4+f} was 17, 
whereas the median size of the 13 benchmarks solved only by {\bf cvc4+mace} 
was 10.
This suggests that  for larger cardinalities {\bf cvc4+mace} suffers 
from the model symmetries created by the introduction of domain constants, 
something that {\bf cvc4+f} avoids.

%------------------------------------------------------------------------------
\subsection{Finite Model Finder Evaluation}\label{sec:results-fmf}
%------------------------------------------------------------------------------

We provide results on \cvc with finite model finding 
for three sets of benchmarks coming from different formal methods applications,
including verification and automated theorem proving.
We will refer to various configurations of \cvc
based on the features they include.
Configuration {\bf cvc4+f} uses the finite model finding techniques described earlier.
Additionally, configurations containing {\bf m} in their suffix use the model-based quantifier instantiation
heuristic described in Section~\ref{sec:fmf-mbqi},
and configurations with {\bf i} use heuristic instantiation,
which can be paired with finite model finding configurations as described in Section~\ref{sec:heur-inst}.

Experiments from Section~\ref{sec:exp-intel}
were run on a Linux machine with an 8-core 2.60GHz 
Intel\textsuperscript{\textregistered} Xeon\textsuperscript{\textregistered} E5-2670 processor.
All others were run on a Linux machine with an 8-core 3.20GHz 
Intel\textsuperscript{\textregistered} Xeon\textsuperscript{\textregistered} E5-1650 processor with 16GB of RAM.

\subsubsection*{Intel benchmarks}
%\label{sec:exp-intel}
We evaluated the overall effectiveness of \cvc's finite model finder 
for quantified SMT formulas taken from verification conditions generated 
by DVF~\cite{GKLT12}, a tool used at Intel for verifying properties 
of security protocols and design architectures, among other applications.
Both unsatisfiable and satisfiable benchmarks were produced, 
the latter by manually removing necessary assumptions 
from verification conditions.
All benchmarks contain quantifiers, although only over uninterpreted sorts, and
span a wide range of theories, 
including linear integer arithmetic, arrays, EUF, and algebraic datatypes.

For comparison we looked at the SMT solvers \cvciii~\cite{BT07} (version 2.4.1)\footnote{\cvciii is the predecessor of \cvc.  
The latter was developed from scratch, and does not have any code in common with \cvciii.  
}, 
Yices~\cite{dutertre2006yices} (version 1.0.32), and \ziii~\cite{DeMoura:2008:ZES:1792734.1792766} (version 4.1).
We did not consider traditional theorem provers and finite model finders
because they do not have built-in support for the theories in our benchmark set.
All these solvers use E-matching as a heuristic method 
for answering unsatisfiable in the presence of universally quantified formulas.
Z3 additionally relies on model-based quantifier instantiation
techniques to establish satisfiability in the presence of quantified formulas~\cite{GeDeM-CAV-09}.

The results, separated into unsatisfiable and satisfiable instances, are shown
in Figure~\ref{fig:dvf-results} for five classes of benchmarks and 
a timeout of 600 seconds per benchmark.
The first two classes, {\bf refcount} and {\bf german}, represent 
verification conditions for systems described in~\cite{GKLT12}; 
benchmarks in the third are taken from~\cite{TG12}; 
the last two classes are verification problems internal to Intel.

\begin{figure}[t]
\scriptsize
\centering
\begin{tabular}{|l|c|r|c|r|c|r|c|r|c|r|}
\hline
Sat & \multicolumn{2}{c|}{\bf german} & \multicolumn{2}{c|}{\bf refcount} & \multicolumn{2}{c|}{\bf agree} & \multicolumn{2}{c|}{\bf apg} & \multicolumn{2}{c|}{\bf bmk}
\\
 & \multicolumn{2}{c|}{(45)} & \multicolumn{2}{c|}{(6)} & \multicolumn{2}{c|}{(42)} & \multicolumn{2}{c|}{(19)} & \multicolumn{2}{c|}{(37)}
\\
\hline
  & \# & time & \# & time & \# & time & \# & time & \# & time
\\
\hline
{\bf cvc3} & 0 & 0.0 & 0 & 0.0 & 0 & 0.0 & 0 & 0.0 & 0 & 0.0
\\
{\bf yices} & 2 & 0.0 & 0 & 0.0 & 0 & 0.0 & 0 & 0.0 & 0 & 0.0
\\
{\bf z3} & 45 & 1.1 & 1 & 7.0 & 0 & 0.0 & 0 & 0.0 & 0 & 0.0
\\
{\bf cvc4+i} & 2 & 0.0 & 0 & 0.0 & 0 & 0.0 & 0 & 0.0 & 0 & 0.0
\\
{\bf cvc4+f } & {\bf 45} & 0.3 & {\bf 6} & 0.1 & 42 & 15.5 & 18 & 200.0 & 36 & 1201.5
\\
{\bf cvc4+f-r } & {\bf 45} & 0.3 & {\bf 6} & 0.1 & 42 & 18.6 & 15 & 364.3 & 34 & 720.4
\\
{\bf cvc4+fi } & 45 & 0.4 & {\bf 6} & 0.1 & {\bf 42} & 14.2 & 19 & 492.8 & 36 & 831.0
\\
{\bf cvc4+fm } & {\bf 45} & 0.3 & {\bf 6} & 0.1 & 42 & 23.6 & {\bf 19} & 210.2 & 37 & 375.1
\\
{\bf cvc4+fmi } & {\bf 45} & 0.3 & {\bf 6} & 0.1 & 42 & 16.4 & 19 & 221.1 & {\bf 37} & 176.8
\\
\hline
%\end{tabular}
%
%\medskip
%
%\begin{tabular}{|l|c|r|c|r|c|r|c|r|c|r|}
\multicolumn{11}{c}{}
\\[-1ex]
\hline
Unsat & \multicolumn{2}{c|}{\bf german} & \multicolumn{2}{c|}{\bf refcount} & \multicolumn{2}{c|}{\bf agree} & \multicolumn{2}{c|}{\bf apg} & \multicolumn{2}{c|}{\bf bmk}
\\
 & \multicolumn{2}{c|}{(145)} & \multicolumn{2}{c|}{(40)} & \multicolumn{2}{c|}{(488)} & \multicolumn{2}{c|}{(304)} & \multicolumn{2}{c|}{(244)}
\\
\hline
 & \# & time & \# & time & \# & time & \# & time & \# & time
\\
\hline
{\bf cvc3} & 145 & 0.4 & 40 & 0.2 & 457 & 6.8 & 267 & 77.0 & 229 & 76.2
\\
{\bf yices} & 145 & 1.8 & 40 & 7.0 & 488 & 1475.4 & 304 & 35.8 & 244 & 25.3
\\
{\bf z3} & 145 & 1.9 & 40 & 0.9 & {\bf 488} & 10.6 & 304 & 12.2 & 244 & 5.3
\\
{\bf cvc4+i} & {\bf 145} & 0.1 & 40 & 0.2 & 484 & 6.8 & {\bf 304} & 11.2 & {\bf 244} & 2.9
\\
{\bf cvc4+f} & 145 & 0.8 & 40 & 0.4 & 476 & 3782.1 & 298 & 2252.5 & 242 & 1507.0
\\
{\bf cvc4+f-r} & 145 & 0.4 & 40 & 0.2 & 475 & 1574.3 & 294 & 3836.0 & 240 & 1930.5
\\
{\bf cvc4+fi } & 145 & 0.7 & {\bf 40} & 0.1 & 488 & 188.7 & 302 & 342.0 & 244 & 660.3
\\
{\bf cvc4+fm } & 145 & 0.4 & 40 & 0.3 & 471 & 5018.2 & 300 & 1122.7 & 242 & 834.1
\\
{\bf cvc4+fmi } & 145 & 0.3 & {\bf 40} & 0.1 & 488 & 185.9 & 302 & 339.8 & 244 & 668.5
\\
\hline
\end{tabular}
\caption{Number of solved satisfiable and unsatisfiable Intel (DVF) benchmarks and cumulative time for solved benchmarks. All times are in seconds.}
\label{fig:dvf-results}
\end{figure}

For the satisfiable benchmarks, our finite model finder is the only tool capable 
of solving any instance in the last three benchmark classes.
In fact, {\bf cvc4+f} is able to solve all but two, and 
most of them in less than a second.
When extended to include techniques for model-based quantifier instantiation (configurations {\bf cvc4+fm} and {\bf cvc4+fmi}),
we are able to solve all satisfiable benchmarks within the timeout.
By comparing {\bf cvc4+f} against {\bf cvc4+f-r}, we see that the region-based approach for recognizing cliques is beneficial, 
particularly for the harder classes where the latter configuration solves fewer benchmarks within the timeout.
The model sizes found for these benchmarks were relatively small; 
only a handful had a model with sort cardinalities larger than 4.
To our knowledge, our model finder is the only tool capable of solving these benchmarks.

For the unsatisfiable benchmarks, Yices and Z3 can solve all of them,
 with Z3 being much faster in some cases.
We observe that \cvc with finite model finding is orders of magnitude slower 
than the SMT solvers on these benchmarks.
This is, however, to be expected since it is geared 
towards finding models, and applies exhaustive instantiation 
with increasingly large cardinality bounds,
which normally delays the discovery that 
the problem is unsatisfiable regardless of those bounds.

However, we found that each unsatisfiable problem can be solved 
by either {\bf cvc4} or {\bf cvc4+fmi}, and in less than 3s.
Additionally, configuration {\bf cvc4+fmi} solves all unsatisfiable benchmarks 
within 900s, suggesting that \cvc's model finder makes consistent progress
towards answering unsatisfiable on provable DVF verification conditions.
From the perspective of verification tools, the results here seem promising.
A common strategy for handling a verification condition would be 
to first use an SMT solver hoping that it can quickly find it unsatisfiable 
with $E$-matching techniques;
and then resort to finite model finding if needed 
to either answer unsatisfiable, or produce a model representing 
a concrete counterexample for the verification condition.
%An even further strategy would be to run the two strategies in parallel.

\subsubsection*{TPTP benchmarks}

We considered benchmarks from a recent version of the TPTP library~\cite{Sut-JAR-09} (5.4.0), 
a widely-used library from the automated theorem proving community.
The benchmarks from this library involve no theory reasoning other than equality,
and are composed mostly of quantified formulas.

We compared \cvc (version 1.2) against other SMT solvers including \ziii (version 4.3) and \cvciii (version 2.4.1), 
as well as various automated theorem provers and model finders for first order logic, 
including Paradox~\cite{Claessen:Soerensson:MACEimprove:ModelComputationWS:2003} and 
iProver~\cite{Kor08-IJCAR} (version 0.99).
Paradox is a MACE-style model finder that uses preprocessing optimizations 
such as sort inference and clause splitting, among others, and
then encodes to SAT the original problem together with increasingly 
looser constraints on the size of the model.
iProver is an automated theorem prover based in the Inst-Gen calculus
that can also run in finite model finding mode ({\bf iprover+f}).
In that mode, it incrementally bounds model sizes 
in a manner similar to MACE-style model finding.
However, it encodes the whole problem into the EPR fragment,\footnote{
This fragment of first-order logic consists of all formulas of the form 
$\exists \tupl{x}. \forall \tupl{y}. \varphi$, 
where $\varphi$ is quantifier-free and contains no function symbols.}
for which it is a decision procedure.
Since these two tools are limited to classical first-order logic with equality, 
we considered only the unsorted first-order benchmarks of TPTP.

\begin{figure}[t]
\tiny
\centering
\begin{tabular}{|l|c|c|c|c|c|c|c|c|c|c|}
\hline
 & \multicolumn{5}{c|}{Unsat} & \multicolumn{5}{c|}{Sat}
\\
\hline
 & {\bf EPR} & {\bf NEQ} & {\bf SEQ} & {\bf PEQ} & {\bf Total} & {\bf EPR} & {\bf NEQ} & {\bf SEQ} & {\bf PEQ} & {\bf Total}
\\
 & (920) & (2008) & (7682) & (1796) & (12406) & (388) & (618) & (340) & (612) & (1958)
\\
\hline
{\bf z3 } &  840 & 1406 & {\bf 3366} & 656 & 6268 & 345 & 261 & 175 & 160 & 941
\\
{\bf cvc3 } &  596 & 910 & 3091 & 648 & 5245 & 24 & 0 & 8 & 0 & 32
\\
{\bf iprover } &  {\bf 888} & {\bf 1786} & 3346 & 310 & {\bf 6330} & {\bf 384} & 434 & 106 & 156 & 1080
\\
{\bf iprover+f } &  - & - & - & - & - & 378 & {\bf 555} & {\bf 224} & 268 & 1425
\\
{\bf paradox } &  - & - & - & - & - & 343 & 534 & 201 & {\bf 372} & {\bf 1450}
\\
{\bf cvc4+i } &  809 & 1346 & 3277 & {\bf 668} & 6100 & 21 & 1 & 8 & 0 & 30
\\
{\bf cvc4+f } &  736 & 900 & 1261 & 531 & 3428 & 329 & 441 & 178 & 242 & 1190
\\
{\bf cvc4+fm } &  725 & 942 & 1315 & 419 & 3401 & 329 & 448 & 214 & 286 & 1277
\\
{\bf cvc4+fi } &  733 & 994 & 1594 & 457 & 3778 & 329 & 422 & 178 & 231 & 1160
\\
{\bf cvc4+fmi } &  748 & 997 & 1594 & 459 & 3798 & 327 & 416 & 190 & 232 & 1165
\\
\hline
\end{tabular}
\caption{Number of solved TPTP benchmarks.  All experiments were run with a 300 second timeout.}
\label{fig:tptp-results}
\end{figure}

Figure~\ref{fig:tptp-results} shows results
for benchmarks from the TPTP library that are known to be satisfiable or unsatisfiable.
All experiments were run with a 300 second timeout per benchmark.
The benchmarks were placed into (exactly one) category based on its logical and syntactic characteristics,
where EPR includes benchmarks that reside in the effectively propositional fragment,
NEQ are benchmarks that do not contain any equality reasoning,
SEQ are benchmarks containing some equality,
and PEQ are benchmarks containing only pure equality.
Both configurations {\bf iprover} and {\bf iprover+f}
used scheduling strategies that iProver incorporated for CASC 24, a competition for automated theorem provers,
meaning that multiple configurations of this solver were run sequentially.
The latter of these configurations, as well as the configuration {\bf paradox} were solely run on the satisfiable benchmarks from this set.
All configurations of \cvc with finite model finding used sort inference techniques as described in~\cite{reynolds2013finite},
which is capable of treating unsorted inputs as multi-sorted based on their structure.
Sort inference techniques are known to be useful for this set of benchmarks, 
and are used in most competitive ATP systems, including Paradox and iProver.

For satisfiable benchmarks,
\cvc's model finder with exhaustive instantiation ({\bf cvc4+f})
solves 1190 benchmarks.
Using model-based quantifier instantiation,
that number goes up to 1277 (configuration {\bf cvc4+fm}).
Using heuristic instantiation ({\bf cvc4+fmi}) in addition to model-based instantiation
led to finding fewer satisfiable benchmarks, solving 1165 within the timeout,
suggesting that the solver becomes overloaded with the large number of instantiations
produced by exhaustive instantiation.

While \cvc solves more than \ziii, which finds 941 satisfiable benchmarks, 
our model finder still trails the overall performance 
of the other model finders on these problems.
Paradox was the overall best solver, finding 1450 satisfiable benchmarks.
We attribute this to the fact that we have not implemented advanced preprocessing techniques,
such as clause splitting, that have been shown to be critical for finding finite models of TPTP benchmarks.
Nevertheless, \cvc's model finder solves more satisfiable benchmarks (214) than Paradox for classes of problems
having some equality reasoning (SEQ).
Collectively, some configuration of \cvc with finite model finding 
was able to solve 52 satisfiable benchmarks that {\bf paradox} was not able to solve,
and 36 satisfiable benchmarks that {\bf iprover+f} was not able to solve.
%Additionally, some configuration of \cvc with finite model finding
%was able to solve \revise{3} satisfiable benchmarks with 1.0 difficulty rating,
%which means that no known ATP system had solved these problems 
%when version of 5.4.0 of the TPTP library was released (in June of 2012).

Figure~\ref{fig:tptp-results} also shows results for unsatisfiable problems.
Although these results are not comparable to those achieved
by state-of-the-art theorem provers, such as Vampire and E~\cite{DBLP:conf/cav/KovacsV13,schulz2002brainiac}, 
we note that {\bf iprover} solves the most benchmarks, 6330.
Here, {\bf cvc4+fmi} was the best configuration of \cvc with finite model finding,
solving 3781 within the timeout.
While finite model finding configurations solved considerably fewer than using heuristic instantiation alone,
some configuration of \cvc with finite model finding solves 144 unsatisfiable benchmarks
that were unable to be solved by any other solver in these experiments,
including iProver and \ziii.

\begin{figure}[t]
\centering
\includegraphics[width=0.49\textwidth]{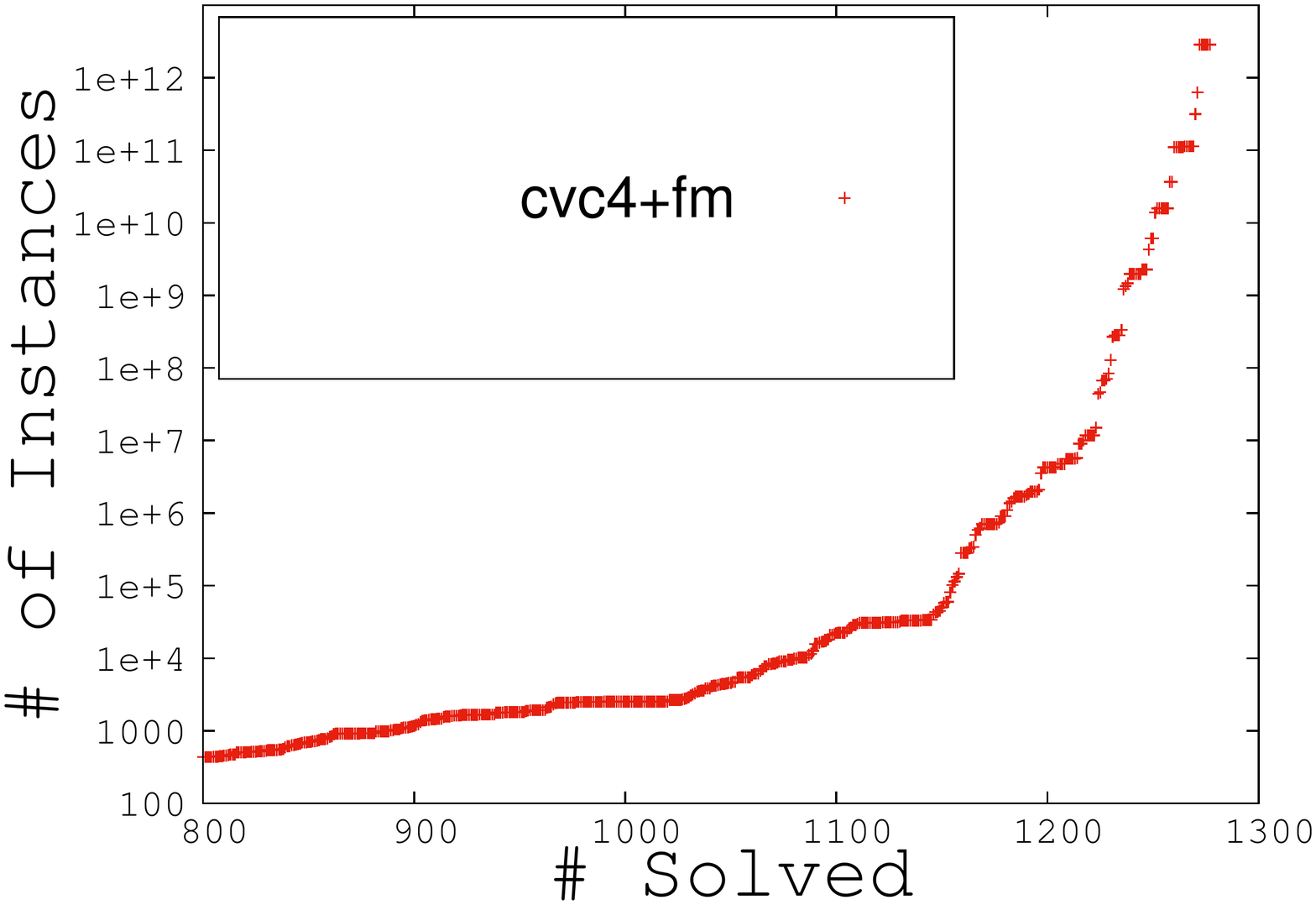}
\includegraphics[width=0.49\textwidth]{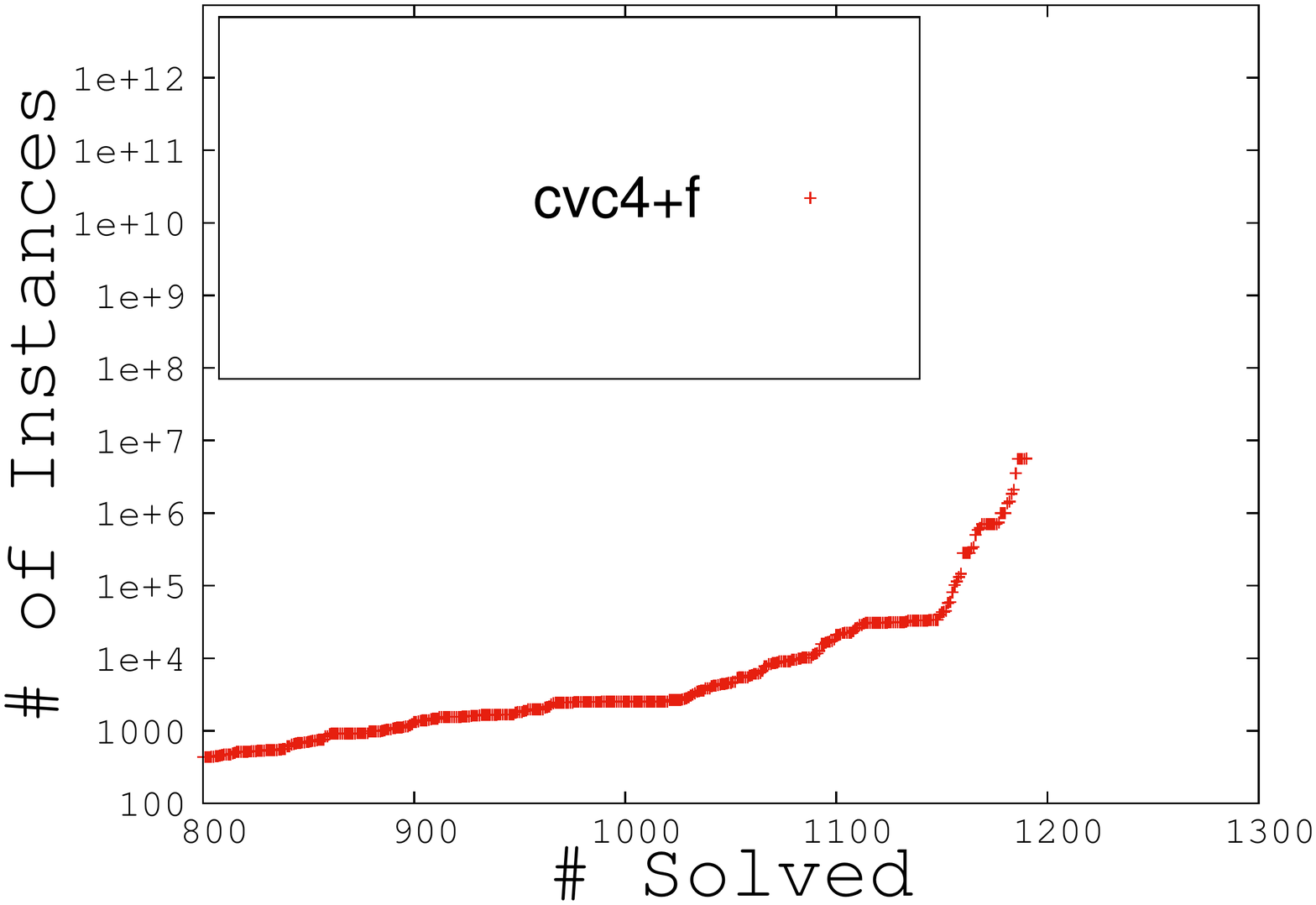}
\caption[Satisfiable TPTP problems with and without model-based instantiation.]{
Satisfiable TPTP problems with ({\bf cvc4+fm}) and without ({\bf cvc4+f}) model-based instantiation.
A point $(x,y)$ on this graph says the configuration solves $x$ benchmarks each having at most $y$ ground instances of quantified formulas.
}
\label{fig:tptp-model-size}
\end{figure}

To further evaluate the impact of model-based quantifier instantiation 
on our model finder, 
we recorded statistics on the domain size of quantified formulas 
in benchmarks solved by its various configurations.
We measured the total number of possible ground instances for all quantified formulas 
in the smallest model for that benchmark (a quantified formula over $n$ variables each with domain size $k$ has $k^n$ instances).
For a problem with $d$ total instances, the configuration {\bf cvc4+f} must explicitly generate these $d$ instances, 
while a model-based configuration may avoid doing so.

The graph on the right hand side of Figure~\ref{fig:tptp-model-size} shows that
{\bf cvc4+f} was only able to solve 13 problems having
more than 100K instances, the maximum having around 5.6 million instances.
On the other hand, {\bf cvc4+fm} was capable of solving 123 problems having 
more than 100K instances, 
with the largest having more than 2.8 trillion instances.
This indicates that the model-based instantiation approach improves the scalability 
of our model finder, and allows it to solve benchmarks 
where exhaustive instantiation is clearly infeasible.
Model finders such as Paradox have other ways of handling the explosion in the number of instances, 
namely by minimizing the number of variables per clause.
Coupling these techniques with model-based techniques could then lead to additional improvements in scalability.
Since techniques for reducing variables in clauses rely on introducing new symbols into the problem,
we have found that they have a negative impact on performance for several classes of benchmarks,
and thus are disabled by default in \cvc.

%Recently, \cvc participated in CASC 24, a competition evaluating the performance of automated theorem provers on selected TPTP benchmarks.
%In the first-order non-theorems division (FNT), 
%\cvc finished $3^{rd}$ out of 8 solvers, solving 96 of 150 problems, just behind Paradox, which solved 99.
%The newest version of iProver won the competition solving 122 problems.
%At the time of the competition (June of 2013), \cvc incorporated improved methods for model-based quantifier instantiation not mentioned in this paper.

\subsubsection*{Isabelle benchmarks}

Recent work has shown that SMT solvers are effective at discharging proof obligations for Isabelle, a generic proof assistant~\cite{paulson2002isabelle}.
The performance of these solvers can benefit from an encoding that makes use of theories~\cite{BBP-11}.
We considered a set of 13,041 benchmarks corresponding to both provable and unprovable proof goals,
corresponding to a superset of those discussed in~\cite{BBP-11}.
Most benchmarks in this set contain quantifiers, and a significant portion contain integer arithmetic.
For many of them, the quantification is limited to the uninterpreted sorts, thus making our finite model finding approach applicable.

\begin{figure}[t]
\centering
{\scriptsize
%\begin{tabular}{|l|c|@{\;}r@{\;}|@{\;}c@{\;}|@{\;}r@{\;}|@{\;}c@{\;}|@{\;}r@{\;}|@{\;}c@{\;}|@{\;}r@{\;}|@{\;}c@{\;}|@{\;}r@{\;}|}
\begin{tabular}{|l|c|c|c|c|c|c|c|c|c|c|}
\hline
Sat & {\bf Arr} & {\bf FFT} & {\bf FTA} & {\bf Hoare} & {\bf NSS} & {\bf QEp} & {\bf SN} & {\bf TSq} & {\bf TSf} & {\bf Total}
\\
\hline
{\bf z3 } & 2 & 19 & 24 & 47 & 7 & 47 & 1 & 17 & 8 & 172 
\\
{\bf cvc3 } & 0 & 9 & 0 & 0 & 0 & 0 & 0 & 8 & 0 & 17 
\\
{\bf cvc4+i } & 0 & 9 & 0 & 0 & 0 & 0 & 0 & 8 & 0 & 17 
\\
{\bf cvc4+f } & 35 & 145 & {\bf 177 } & 162 & 56 & 85 & 12 & {\bf 57 } & 90 & 819 
\\
{\bf cvc4+fm } & 33 & 141 & 173 & 155 & 43 & 86 & 12 & 54 & 89 & 786 
\\
{\bf cvc4+fi } & 36 & 146 & 172 & 162 & {\bf 61 } & {\bf 86 } & {\bf 12 } & 55 & {\bf 93 } & 823 
\\
{\bf cvc4+fmi } & {\bf 36 } & {\bf 147 } & 174 & {\bf 162 } & 61 & 83 & 12 & 56 & 93 & {\bf 824 }
\\
\hline
\multicolumn{10}{c}{}
\\[-1ex]
\hline
Unsat & {\bf Arr} & {\bf FFT} & {\bf FTA} & {\bf Hoare} & {\bf NSS} & {\bf QEp} & {\bf SN} & {\bf TSq} & {\bf TSf} & {\bf Total}
\\
\hline
{\bf z3 } & 178 & 277 & 917 & 549 & 108 & 325 & 241 & 620 & {\bf 291 } & 3506 
\\
{\bf cvc3 } & {\bf 321 } & {\bf 296 } & {\bf 1124 } & {\bf 607 } & 105 & 297 & 207 & 643 & 227 & 3827 
\\
{\bf cvc4+i } & 307 & 288 & 990 & 563 & {\bf 117 } & {\bf 360 } & {\bf 242 } & {\bf 708 } & 283 & {\bf 3858 } 
\\
{\bf cvc4+f } & 165 & 106 & 451 & 239 & 44 & 131 & 88 & 442 & 151 & 1817 
\\
{\bf cvc4+fm } & 132 & 92 & 442 & 238 & 26 & 160 & 88 & 430 & 128 & 1736 
\\
{\bf cvc4+fi } & 172 & 185 & 589 & 383 & 47 & 222 & 112 & 585 & 196 & 2491 
\\
{\bf cvc4+fmi } & 168 & 186 & 589 & 379 & 47 & 222 & 112 & 584 & 196 & 2483 
\\
\hline
\end{tabular}
}
\caption{
Number of solved satisfiable and unsatisfiable Isabelle benchmarks
for various classes within a 300 second timeout.}
\label{fig:isabelle-results}
\end{figure}

The results are shown in Figure~\ref{fig:isabelle-results}.
For satisfiable benchmarks, all configurations of \cvc's model finder find more satisfiable 
problems than \ziii, which finds only 172 of them overall.
The model-based quantifier instantiation technique from Section~\ref{sec:fmf-mbqi} (configuration {\bf cvc4+fm})
was less effective than naive instantiation (configuration {\bf cvc4+f}) which solves 819,
suggesting that model-based techniques were not effective at minimizing the number of instantiations for this set of benchmarks.
Using heuristic E-matching noticeably improved the search for models, 
as configuration {\bf cvc4+fi} solves 823 satisfiable benchmarks.
Using both model-based instantiation and heuristic instantiation, configuration {\bf cvc4+fmi},
found more satisfiable problems (824) than any other configuration.

For unsatisfiable problems, {\bf cvc4+i} is the overall winner, solving 3,858,
which was more than both {\bf z3 } and {\bf cvc3} which solved 3,506 and 3,827 respectively.
Configurations of \cvc with finite model finding generally solves less unsatisfiable benchmarks,
but is orthogonal to other solvers and configurations.
In these experiments, 309 unsatisfiable benchmarks that \cvciii cannot solve are solved by at least one configuration of \cvc with finite model finding.
Similarly, a configuration of \cvc with finite model finding solves 429 unsatisfiable benchmarks that \ziii cannot, and 168 that {\bf cvc4+i} cannot.

\section{Conclusion}
\label{sec:conclusion}

We developed a general approach for finite model finding in SMT that 
is efficient for many classes of problems that are of practical interest to formal methods applications.
Experimental evidence from an implementation of these methods in the SMT solver \cvc
shows that our approach is effective in practice at solving many classes of benchmarks, 
including verification conditions from industry,
and benchmarks from automated theorem proving libraries.
The implementation is highly competitive  with respect to other SMT solvers and
to automated theorem provers.

In ongoing work, 
we plan to extend our approach to the problem of finding models of formulas 
with quantifiers ranging over built-in domains such as the integers and inductive datatypes.
We are also investigating the use of \cvc as a backend to interactive proof assistants
such as Isabelle and Coq, 
where small counterexamples to conjectures are often helpful to the user.

\begin{acknowledgements}
We would like to thank Amit Goel and Sava Krsti\'c for their valuable contributions to this work,
Sascha B\"ohme for proving the Isabelle benchmarks,
and Fran\c{c}ois Bobot for his help in writing a TPTP front end for CVC4.
\end{acknowledgements}

\bibliographystyle{spmpsci}      % mathematics and physical sciences
\bibliography{pap}

\end{document}